\documentclass[a4paper,11pt]{article}

\usepackage[utf8]{inputenc}
\usepackage[T1]{fontenc}
\usepackage{amsmath,amssymb,amsthm}
\usepackage{booktabs}
\usepackage{placeins}
\usepackage[margin=2.5cm]{geometry}
\usepackage{graphicx}
\usepackage{natbib}
\usepackage[hidelinks]{hyperref}

\newtheorem{theorem}{Theorem}
\newtheorem{corollary}[theorem]{Corollary}
\newtheorem{lemma}[theorem]{Lemma}

\title{Where the Quantum Lives in D-Wave Hybrid Portfolio Optimization:\\
An Operational Decomposition Audit}

\author{
Luis Lozano\footnote{e-mail address: {\tt lalozanom@tec.mx}; ORCID:
\href{https://orcid.org/0000-0001-7202-3437}{0000-0001-7202-3437}}
\\[6pt]
{\small \em EGADE Business School, Tecnol\'ogico de Monterrey,}\\
{\small\em Santa Fe, Mexico City, Mexico}
}

\date{}

\begin{document}

\maketitle

\begin{abstract}
Hybrid quantum-classical solvers conceal how reported performance divides
between quantum-processing-unit (QPU) access and other service time.  We
audit D-Wave's Leap service on cardinality-constrained mean-variance
portfolio instances from $N{=}10$ to $640$, comparing constraint-native
constrained quadratic models (CQMs), penalty-encoded binary quadratic models
(BQMs), Gurobi mixed-integer quadratic programming, simulated annealing, and
a matched-budget Tabu baseline.  We report the exposed timing fields and
propose a four-metric operational audit.

LeapHybridCQM matches Gurobi's proven optimum on all $54$ head-to-head
instances at $N\leq120$.  Its mean QPU access is $0.034$ seconds, $0.68\%$
of the nominal five-second budget and approximately $0.72\%$ of measured run
time; telemetry does not resolve the residual into classical computation,
orchestration, or budget granularity.  Across $108$ fixed-seed, five-second
CPU runs on $36$ instances, \texttt{TabuSampler} reaches CQM-level objectives
with mean absolute delta $0.00078$ (maximum $0.00798$); at
$N\in\{400,640\}$, mean absolute delta is $0.0000304$.  This external
comparison establishes classical objective achievability, not an internal
QPU ablation.

We prove that quadratic encoding of exact cardinality adds a dense rank-one
term and, apart from isolated exact coefficient cancellations, yields a
complete logical graph independent of covariance support, with a bounded-degree
embedding lower bound.  The resulting density collapse
is consistent with hybrid-BQM degradation.  A descriptive five-window
Fama--French overlay of portfolios produced by the post-projection
direct-QPU path yields mean
Sharpe $1.94$ versus $2.22$ for $1/N$.

These results do not establish a quantum-sampling advantage.  They motivate
reporting QPU wall-clock fraction, optimality gap, graph-density amplification,
service-output objective variance, and matched-budget classical controls,
while keeping operational accounting separate from causal attribution.
\end{abstract}

\vspace{4pt}

% ============================================================================
\section{Introduction}
\label{sec:intro}
% ============================================================================

Quantum annealing has attracted sustained interest as a potential accelerator for combinatorial
portfolio optimization~\citep{orus2019,herman2023}. The standard approach encodes a
cardinality-constrained mean-variance objective as a quadratic unconstrained binary optimization
(QUBO) problem, adds a quadratic penalty term to enforce the cardinality constraint, embeds the
resulting logical graph onto the hardware topology via minor embedding, and samples on the
quantum processing unit (QPU)~\citep{venturelli2019,mugel2022,sakuler2025}. A parallel line of
work uses D-Wave's hybrid cloud solvers, which accept either penalty-encoded binary quadratic
models (BQMs) or constraint-native constrained quadratic models (CQMs) and allocate classical
and QPU resources internally~\citep{mugel2022}.

Despite this body of work, the relative performance of these formulation choices has not been
systematically characterized across the joint space of problem size, covariance density, and
solver budget. Most studies benchmark a single solver at a single scale. We address this gap by
defining a canonical cardinality-constrained mean-variance (MV) portfolio problem and comparing four solver
paths (direct QPU, hybrid BQM, hybrid CQM, and classical baselines: Gurobi MIQP and simulated
annealing) across controlled synthetic instances and real equity data from the Fama--French 49
industry portfolios~\citep{french_data}.  The formulation permits an optional
turnover term for re-optimisation from a fixed previous portfolio, but every
reported campaign sets that term to zero.  The empirical testbed is therefore
the single-period MV special case; this choice isolates cardinality encoding
and solver-interface effects, and the results do not test non-zero switching
costs.  The intended users of these results are D-Wave
practitioners choosing between solver interfaces, fintech engineering teams evaluating hybrid
solvers for constrained allocation problems, and quantitative researchers benchmarking
quantum-assisted portfolio optimization.

\paragraph{Contributions.}
This paper makes three principal contributions to the literature on hybrid
quantum-classical portfolio optimization.  The first is empirical: on the
$162$ synthetic-hybrid runs of the $N{=}10$--$640$ campaign, the
constraint-native LeapHybridCQM service matches Gurobi's proven optimum on
every head-to-head instance where Gurobi proves optimality
($54$ instances at $N \leq 120$; a paired one-sided Wilcoxon signed-rank test
rejects $H_0:f_{\mathrm{CQM}}\geq f_{\mathrm{BQM}}$ at every $N \geq 20$,
Holm-adjusted $p < 0.05$), and we report all three SDK timing fields rather than the
single $t_{\mathrm{QPU}}$ measurement that has previously been treated as
the proxy for ``quantum work.''  Mean $t_{\mathrm{QPU}}$ is $0.034$
seconds out of the $5$-second hybrid CQM minimum budget, a $0.68\%$
QPU wall-clock fraction, and the saved records contain exactly
identical objective values across all available repeated calls at every tested
$(N, \rho, B)$ cell.\footnote{$10$ repetitions completed in every cell
except $(N{=}50, \text{dense}, 300\,\text{s})$, where four completed
within the campaign wall-clock budget; the objective-stability observation is
verified at the cell level against the saved JSONL records.}

The second contribution is a matched-budget classical-reproducibility
comparison.  Because the protocol's
$r_{\mathrm{QPU}}$ metric (introduced in the third contribution below)
bounds wall-clock share rather than causal contribution, we complement
the telemetry with a CPU-only matched-wall-clock comparison:
\texttt{TabuSampler} 1.7.0 at $5$\,s on the same
penalty-encoded BQM that hybrid BQM consumes, on $36$ instances at
$N \in \{50, 80, 120, 200, 400, 640\}$ with three fixed algorithm seeds
per instance ($108$ runs).  Across the full grid, the mean absolute
Tabu--CQM delta is $0.00078$ (maximum $0.00798$); at
$N \in \{400, 640\}$, the scales with the largest CQM-over-BQM
margins and no optimality anchor, it is $0.000030$ (maximum
$0.000989$).
Because the penalty-encoded BQM and the constraint-native CQM encode
the same cardinality-constrained problem (same feasible set, same
optimum), matching the CQM objective on the BQM matches the CQM path's
solution quality; equal objective value does not by itself imply the
identical selected portfolio unless the optimizer is unique, which we
do not claim.  Combined with the $0.68\%$ wall-clock fraction, this
shows that, at matched compute, a classical heuristic alone reaches
the same objective levels at every tested scale; this is an
objective-achievability result, not an internal ablation of the
hybrid pipeline.

The third contribution is methodological.  We formalise the empirical
density-axis collapse as a theorem: for a cardinality constraint
$\mathbf{1}^\top z = K$ enforced by quadratic penalty $A \cdot
(\mathbf{1}^\top z - K)^2$, the encoded QUBO matrix has off-diagonal
support equal to the complete graph $K_N$ regardless of the input
covariance support, except at the finite set of exact coefficient-cancellation
values identified in Theorem~\ref{thm:density_collapse}; none occurs in the
tested encodings.  A corollary gives an $\Omega(N^2/d)$ physical-qubit
lower bound on bounded-degree hardware, and a conditional lemma links an
individual chain's length to its break probability.  Building on these
structural results, we
propose an operational decomposition audit protocol consisting of four
metrics (the QPU wall-clock fraction $r_{\mathrm{QPU}}$, the optimality
gap $g_F$, the encoded-graph density amplification $d_{\mathrm{amp}}$, and
the service variance $v_{\mathrm{service}}$) and a seven-step reporting
procedure for hybrid benchmarks.  A reference implementation accompanies
the paper.  An out-of-sample financial layer on Fama--French
$49$-industry portfolios further shows that portfolios retained after the
direct-QPU path's exact-$K$ projection marginally underperform a naive $1/N$
benchmark in Sharpe terms ($-0.28$ on average across five rolling
windows), with the largest positive QPU-minus-baseline spread occurring
in the earliest sampled window.

Together, these results quantify what a reported D-Wave hybrid win on
this problem class can claim: on the exposed telemetry it presents as a
constraint-native, classically orchestrated pipeline, consistent with
the service's documented design~\citep{dwave_hybrid_solvers_whitepaper},
with a small measured QPU-access component, and the operational
benchmarking standard should report all three timing fields and explicit
matched-compute classical baselines rather than $t_{\mathrm{QPU}}$
alone.  Throughout, ``where the quantum lives'' refers to the
operational, wall-clock location of QPU activity in the service's
telemetry; we make no claim to localize a causal quantum contribution,
which the exposed telemetry cannot identify.

In benchmarking terms, this is a system audit rather
than a model-independent application benchmark.  Application-oriented
benchmark suites emphasize end-to-end quality and time-to-solution
reporting~\citep{lubinski2023}; robust optimization benchmarks additionally
require formulation-appropriate classical and quantum baselines and equal
hyperparameter effort~\citep{bucher2024}.  The broader standardization
framework of~\citet{acuaviva2026} identifies relevance, reproducibility,
fairness, verifiability, and usability as benchmark quality attributes.
Most recently, \citet{koch2026qoblib} distinguish application, algorithm,
and system benchmarking and note that only model-independent application
benchmarks can support a quantum-advantage claim.  We adopt that taxonomy:
the present study cannot establish quantum advantage.  Its narrower
contribution is to operationalize resource accounting for an opaque hybrid
service by combining exposed timing fields, formulation-level graph
amplification, output variance, and matched-budget classical reproducibility.

This study sits next to several portfolio-optimization benchmarks using quantum and hybrid methods.
\citet{mugel2022} study dynamic portfolio optimization on real data with quantum and
quantum-inspired methods. \citet{sakuler2025} compare QBSolv, Hybrid BQM, Hybrid CQM, and exact
classical solutions for a production portfolio-optimization use case, finding that CQM is useful for
handling constraints without manual penalty tuning. \citet{lang2022} and \citet{buonaiuto2023}
benchmark annealing and gate-model workflows for portfolio optimization, while \citet{acharya2025}
develop a decomposition pipeline for large constrained portfolio and rebalancing problems.
\citet{morapakula2025endtoend} build an end-to-end deployable pipeline that combines CQM-based
discrete asset selection with classical weight allocation and rebalancing, demonstrating the
operational viability of hybrid annealing in a real market context.  Our contribution is
narrower and more diagnostic: we isolate the formulation choice for a cardinality-constrained mean-variance
problem, sweep size, covariance density, and solver budget, and connect the hybrid BQM/CQM
comparison to the submitted graph's density amplification and direct-QPU
embedding overhead without treating either as an internal hybrid-pipeline
ablation.

We note explicitly that this is a formulation-and-platform study within the D-Wave ecosystem,
not a broad quantum-vs-classical superiority claim. Classical solvers (Gurobi MIQP) find
provably optimal solutions on $62$ of $63$ attempted instances at $N \leq 200$
(including all $54$ head-to-head instances at $N \leq 120$; one $N{=}200$
run reaches the time limit), while the academic license does not support the
model size at $N \geq 400$; the practical question addressed here is
which D-Wave solver interface to use when a practitioner has chosen to use D-Wave hardware.

My companion paper~\citep{lozano2026penalty} addresses
the adjacent practitioner question of how to make direct-QPU portfolio
optimisation viable by replacing the penalty-encoded QUBO with an
objective-only QUBO and a classical feasibility projector.  The
density-collapse root diagnosis (Theorem~\ref{thm:density_collapse}
below) is common to both; the present paper restricts attention to the
existing LeapHybridCQM and LeapHybridBQM service interfaces, while the
companion paper builds a custom direct-QPU pipeline.  Each paper stands
alone.

% ============================================================================
\section{Background}
\label{sec:background}
% ============================================================================

\subsection{Cardinality-Constrained Portfolio Selection}

The classical mean-variance portfolio problem~\citep{markowitz1952} selects asset weights to
maximize expected return for a given risk level. Adding a cardinality constraint (select
exactly $K$ out of $N$ candidate assets) makes the problem NP-hard in
general~\citep{bertsimas2009cardinality,lucas2014} and motivates a long
line of mixed-integer quadratic-programming formulations in the OR /
finance literature~\citep{bertsimas2009cardinality}. In the binary formulation, each asset $i$ has a selection variable
$z_i \in \{0, 1\}$ and the objective is:
\begin{equation}
    \min_{z \in \{0,1\}^N} \; -\boldsymbol{\mu}^\top z + \lambda \, z^\top \Sigma \, z
    \quad \text{s.t.} \quad \mathbf{1}^\top z = K,
    \label{eq:basic}
\end{equation}
where $\boldsymbol{\mu}$ is the expected return vector, $\Sigma$ is the covariance matrix, and
$\lambda > 0$ is the risk-aversion parameter.

\subsection{Optional Mean-Variance-Turnover Extension}

For completeness, repeated re-optimization across consecutive decision
periods can extend~\eqref{eq:basic} with a switching-cost penalty:
\begin{equation}
    \min_{z \in \{0,1\}^N} \; -\boldsymbol{\mu}^\top z + \lambda \, z^\top \Sigma \, z
    + \boldsymbol{\tau}^\top |z - z_{t-1}|
    \quad \text{s.t.} \quad \mathbf{1}^\top z = K,
    \label{eq:mvt}
\end{equation}
where $z_{t-1}$ is the previous binary portfolio and $\boldsymbol{\tau}$ is the per-asset
switching cost. Because both $z$ and $z_{t-1}$ are binary, the absolute-value term is
linear: $|z_i - z_{t-1,i}| = z_i(1 - z_{t-1,i}) + (1 - z_i)z_{t-1,i}$, which absorbs into
the QUBO diagonal without introducing auxiliary variables or additional pairwise couplers.
All experiments reported below set $\boldsymbol{\tau}=\mathbf{0}$ and do
not pass a previous portfolio, so their empirical objective is the
single-period MV objective~\eqref{eq:basic}.  Equation~\eqref{eq:mvt}
establishes that a fixed-previous-portfolio turnover extension would not
change the graph-support argument; non-zero turnover is not benchmarked here.

\subsection{Penalty-Encoded QUBO (BQM Path)}

The standard QUBO encoding promotes the cardinality constraint into the objective via a
quadratic penalty:
\begin{equation}
    Q_{\text{BQM}} = -\text{diag}(\boldsymbol{\mu}) + \lambda \Sigma
    + A \cdot \mathbf{1}\mathbf{1}^\top - 2AK \cdot I
    + \text{diag}(\boldsymbol{\tau}_{\text{shift}}),
    \label{eq:qubo}
\end{equation}
where $A > 0$ is the penalty weight and $\boldsymbol{\tau}_{\text{shift},i} = \tau_i(1 -
2z_{t-1,i})$. At the tested penalty scale, the $A \cdot
\mathbf{1}\mathbf{1}^\top$ term makes the QUBO fully connected
regardless of the sparsity of $\Sigma$, a phenomenon termed \emph{constraint
dilution}~\citep{verma2020,lozano2026penalty}. This full connectivity forces $N(N-1)/2$
couplers in the logical graph and long minor-embedding chains on hardware topologies with
bounded degree.

\subsection{Constraint-Native CQM (CQM Path)}

D-Wave's constrained quadratic model (CQM) interface accepts the cardinality constraint
natively, eliminating the penalty term entirely:
\begin{equation}
    Q_{\text{CQM}} = -\text{diag}(\boldsymbol{\mu}) + \lambda \Sigma
    + \text{diag}(\boldsymbol{\tau}_{\text{shift}}),
    \quad \text{constraint:} \; \mathbf{1}^\top z = K.
    \label{eq:cqm}
\end{equation}
The CQM objective preserves the sparsity structure of $\Sigma$. The hybrid solver handles
constraint satisfaction internally through a combination of classical and quantum resources.

\subsection{Density-Axis Collapse}
\label{sec:density_collapse}

A structural consequence of~\eqref{eq:qubo} is that the original covariance
density $\rho$ of $\Sigma$ is irrelevant to the encoded logical graph density
when using penalty encoding: except when an original coefficient exactly
cancels the penalty coefficient, the $A \cdot
\mathbf{1}\mathbf{1}^\top$ term produces a complete graph.  No such
cancellation occurs in the tested encodings.  We formalise the general
statement and its finite exception set as a theorem.

\begin{theorem}[Density-Axis Collapse under Quadratic Penalty Encoding]
\label{thm:density_collapse}
Adopt the symmetric QUBO convention $H(z) = z^\top Q z$ with $Q = Q^\top$;
under this convention the off-diagonal coefficient of $z_i z_j$ ($i \neq j$)
in the Hamiltonian is $2 Q_{ij}$.  Let $z \in \{0,1\}^N$ and consider a
linear equality constraint $\sum_{i=1}^N a_i z_i = b$ with
$a \in \mathbb{R}^N$ and $b \in \mathbb{R}$, enforced in QUBO form by the
quadratic penalty $A \cdot (\sum_i a_i z_i - b)^2$ added to an objective
with symmetric cost matrix $C \in \mathbb{R}^{N \times N}$.  Define the
\emph{off-diagonal support} of a matrix $M$ as
$\mathrm{supp_{off}}(M) := \{(i,j) : i < j,\, M_{ij} \neq 0\}$.  Then for
every penalty weight $A$ outside the finite set
$\{-C_{ij}/(a_i a_j) : i < j,\, a_i a_j \neq 0\}$,
\begin{equation}
  \mathrm{supp_{off}}(Q)
  \;=\;
  \mathrm{supp_{off}}(C)
  \;\cup\;
  \{(i,j) : i < j,\, a_i a_j \neq 0\}.
  \label{eq:supp-off-Q}
\end{equation}
In particular, if $a_i \neq 0$ for all $i$ (e.g., $a = \mathbf{1}$ for an
exact-cardinality constraint), then $\mathrm{supp_{off}}(Q) =
\{(i,j) : 1 \leq i < j \leq N\}$ and the \emph{logical graph} of $Q$ ---
defined as the graph on $[N]$ with edge set
$\{\{i,j\} : i \neq j,\, Q_{ij} \neq 0\}$ --- is the complete graph $K_N$.
\end{theorem}

\begin{proof}
Expanding the penalty,
\begin{align*}
A \cdot (a^\top z - b)^2
   &= A \cdot z^\top (a a^\top) z \;-\; 2Ab \cdot a^\top z \;+\; A b^2.
\end{align*}
The first term contributes $A \cdot a_i a_j$ to the off-diagonal entry
$Q_{ij}$ for every $i \neq j$ with $a_i a_j \neq 0$; equivalently, it
contributes $2 A \cdot a_i a_j$ to the coefficient of $z_i z_j$ in an
upper-triangular QUBO convention.  The first term also contributes
$A \cdot a_i^2$ to $Q_{ii}$.  The second is linear in $z$ and
contributes $-2Ab \cdot a_i$ to $Q_{ii}$.  The constant $Ab^2$ shifts
the total energy.  Combining with the original off-diagonal contribution
from $C$, we obtain $Q_{ij} = C_{ij} + A \cdot a_i a_j$ for $i \neq j$;
this is non-zero whenever $a_i a_j \neq 0$ and $A \neq -C_{ij}/(a_i a_j)$.
The exception set, $\{-C_{ij}/(a_i a_j) : i < j,\, a_i a_j \neq 0\}$, is
finite (at most $\binom{N}{2}$ values).  For $A$ outside this set,
equation~\eqref{eq:supp-off-Q} holds entry-by-entry.  When $a_i \neq 0$
for all $i$, the second set on the right-hand side equals
$\{(i,j) : i < j\}$, which yields the logical-graph claim.
\end{proof}

\medskip\noindent\emph{In plain terms.}  The cardinality penalty forces
every asset to interact with every other asset through the shared
quadratic term, so the encoded problem becomes fully connected even when
the original covariance matrix is diagonal.  For the cardinality
constraint of Section~\ref{sec:background}, $a = \mathbf{1}$ and $b = K$,
and the second set on the right of equation~\eqref{eq:supp-off-Q}
becomes $\{(i,j) : i < j\}$ regardless of $\mathrm{supp_{off}}(\Sigma)$.

\medskip\noindent\textbf{Remark} (MVT objective inherits the same collapse).
The mean-variance-turnover (MVT) objective of~\eqref{eq:mvt} differs
from the bare MV objective by a turnover term
$\boldsymbol{\tau}^\top |z - z_{t-1}|$.  For binary
$z, z_{t-1} \in \{0,1\}^N$, the identity
$|z_i - z_{t-1,i}| = z_{t-1,i} + (1 - 2 z_{t-1,i}) z_i$ holds by
enumeration, so the turnover term is linear in $z$ and contributes only
to the diagonal of $Q$.  Consequently $\mathrm{supp_{off}}(Q^{\mathrm{MVT}})
= \mathrm{supp_{off}}(Q^{\mathrm{MV}})$, and
Theorem~\ref{thm:density_collapse} applies to MVT verbatim.

\medskip\noindent\textbf{Remark} (CQM path preserves $\mathrm{supp_{off}}(\Sigma)$).
\label{rem:cqm_preserves}
The constraint-native CQM formulation in~\eqref{eq:cqm} contains no
penalty term; the cardinality constraint is held natively by the
solver and does not modify the objective coefficients.  Hence the
off-diagonal quadratic objective support is
$\mathrm{supp_{off}}(Q_{\mathrm{CQM}}) = \mathrm{supp_{off}}(\Sigma)$,
with no $A \cdot \mathbf{1}\mathbf{1}^\top$ contribution.

\begin{corollary}[Embedding-Overhead Lower Bound for $K_N$ on a Bounded-Degree Topology]
\label{cor:embedding_lower_bound}
Let $G = (V_G, E_G)$ be a hardware topology with maximum vertex degree
$d$.  Any minor embedding of $K_N$ into $G$ uses at least
$N(N-1) / d$ physical qubits in total across all chains:
\[
  L_{\mathrm{tot}} \;:=\; \sum_{i=1}^N |\mathrm{chain}_i| \;\geq\; \frac{N(N-1)}{d}.
\]
Equivalently, the per-logical physical-qubit overhead grows at least
as $\Omega(N/d)$.
\end{corollary}

\begin{proof}
$K_N$ has $\binom{N}{2}$ logical edges.  Each logical edge requires at
least one \emph{inter-chain} physical edge --- a physical edge in $E_G$
with one endpoint in $\mathrm{chain}_i$ and the other in
$\mathrm{chain}_j$ for the respective logical qubits $i \neq j$.
Select one such physical edge for each logical edge.  This selected set
has $2\binom{N}{2}=N(N-1)$ endpoints.  Intra-chain edges play no role
in this count and only reduce the available degree budget.  Because each
physical vertex has degree at most $d$, the selected endpoint count is
at most $L_{\mathrm{tot}}d$.  Rearranging gives the claimed bound.
\end{proof}

\medskip\noindent
For Pegasus ($d = 15$) at $N = 80$, integrality strengthens the bound to
$L_{\mathrm{tot}}\geq 422$ physical qubits; the observed Pegasus mean is $769$,
consistent with the bound.  For Zephyr ($d = 20$) at the same $N$ the
bound gives $\approx 316$; the observed Zephyr mean is $731$.  (The
two-topology pooled mean reported in Table~\ref{tab:a1_cbf} is $750$.)  Tighter
topology-specific constants for Pegasus and Zephyr appear
in~\citet{boothby-bunyk-raymond-roy-2020}.

\begin{lemma}[Chain-Length Monotonicity under an Independent-Edge Break Model]
\label{lem:embedding_growth}
Under Theorem~\ref{thm:density_collapse}, the mean chain length for the
penalty-encoded $K_N$ embedding satisfies $\bar L(N) = \Omega(N/d)$
(Corollary~\ref{cor:embedding_lower_bound}).  Consequently, all else
equal --- at fixed chain strength $J_c$, fixed annealing schedule, fixed
calibration, and under a per-edge intra-chain disagreement model in
which each physical edge in a chain has independent break probability
$p$ --- the break probability of an individual chain of length $L$
equals $1 - (1-p)^{L-1}$ and is monotone non-decreasing in $L$.
\end{lemma}

\medskip\noindent
\emph{Scope.}  The lemma is stated per chain.  The \emph{mean} break
probability across an embedding tracks the mean chain length
$\bar L(N)$ only under additional assumptions on the chain-length
distribution (e.g., stochastic ordering of the distributions as $N$
grows); we use the lemma as a heuristic scaling argument connecting
Corollary~\ref{cor:embedding_lower_bound} to the observed chain-break
growth, not as a sharp bound.

\medskip\noindent
\emph{Empirical manifestation.}  We confirm density-axis collapse
empirically in Section~\ref{sec:a1}: chain-break fractions and embedding
overheads are nearly identical across diagonal
($\rho \approx 2/N$, decaying with $N$), block
($\rho \in [0.11, 0.25]$), and dense ($\rho \approx 1.0$) families at
each $N$.  Figure~\ref{fig:density_collapse_schematic} illustrates the
mechanism schematically; Figure~\ref{fig:formulation_choice} summarises
the practical consequences for D-Wave practitioners.

\begin{figure}[t]
    \centering
    \includegraphics[width=0.85\textwidth]{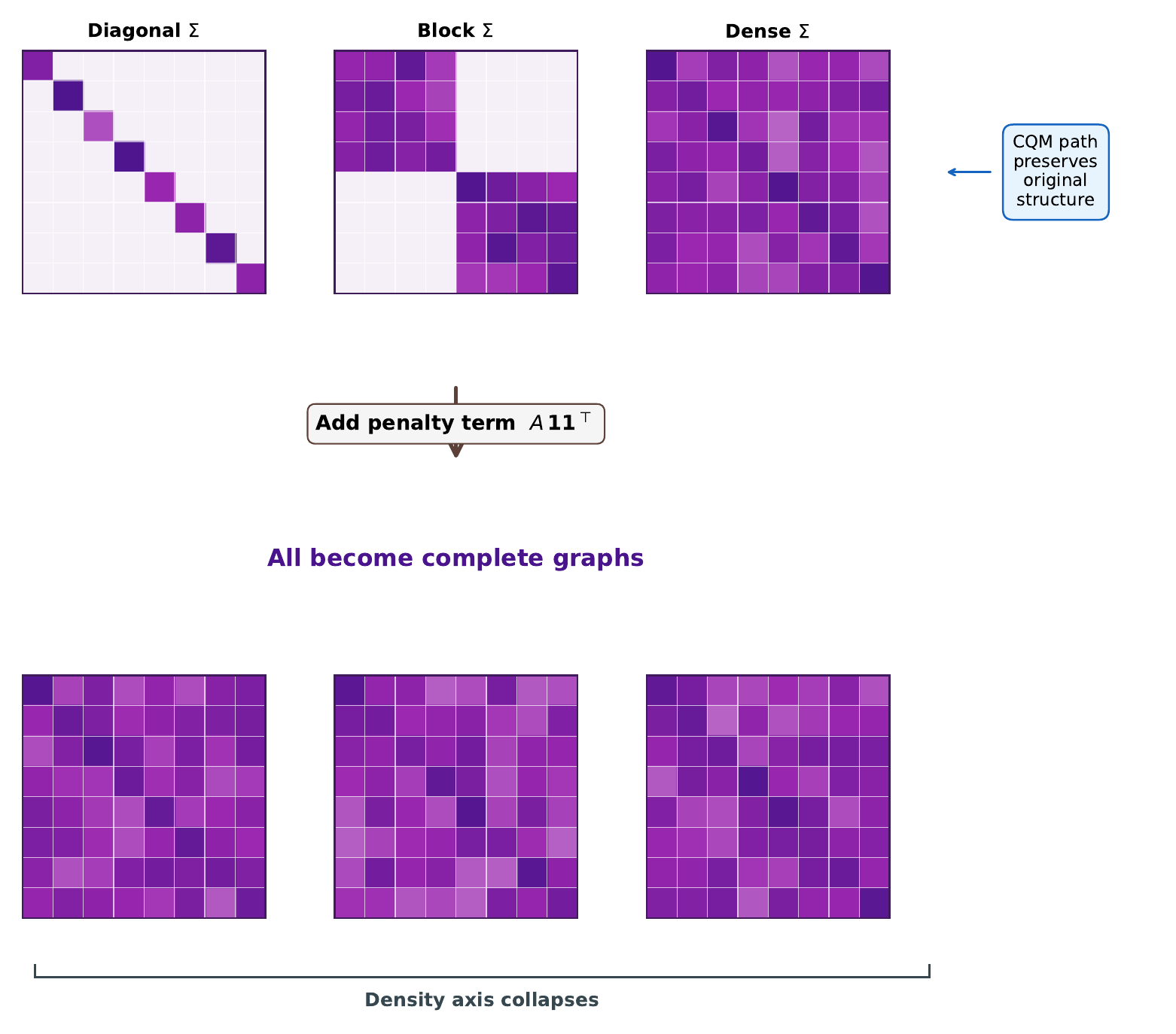}
    \caption{Density-axis collapse under penalty encoding. Three structurally different
    covariance matrices (diagonal, block, dense) all become complete graphs after adding the
    cardinality penalty $A \cdot \mathbf{1}\mathbf{1}^\top$. The CQM path preserves the
    original structure.}
    \label{fig:density_collapse_schematic}
\end{figure}

\begin{figure}[t]
    \centering
    \includegraphics[width=0.75\textwidth]{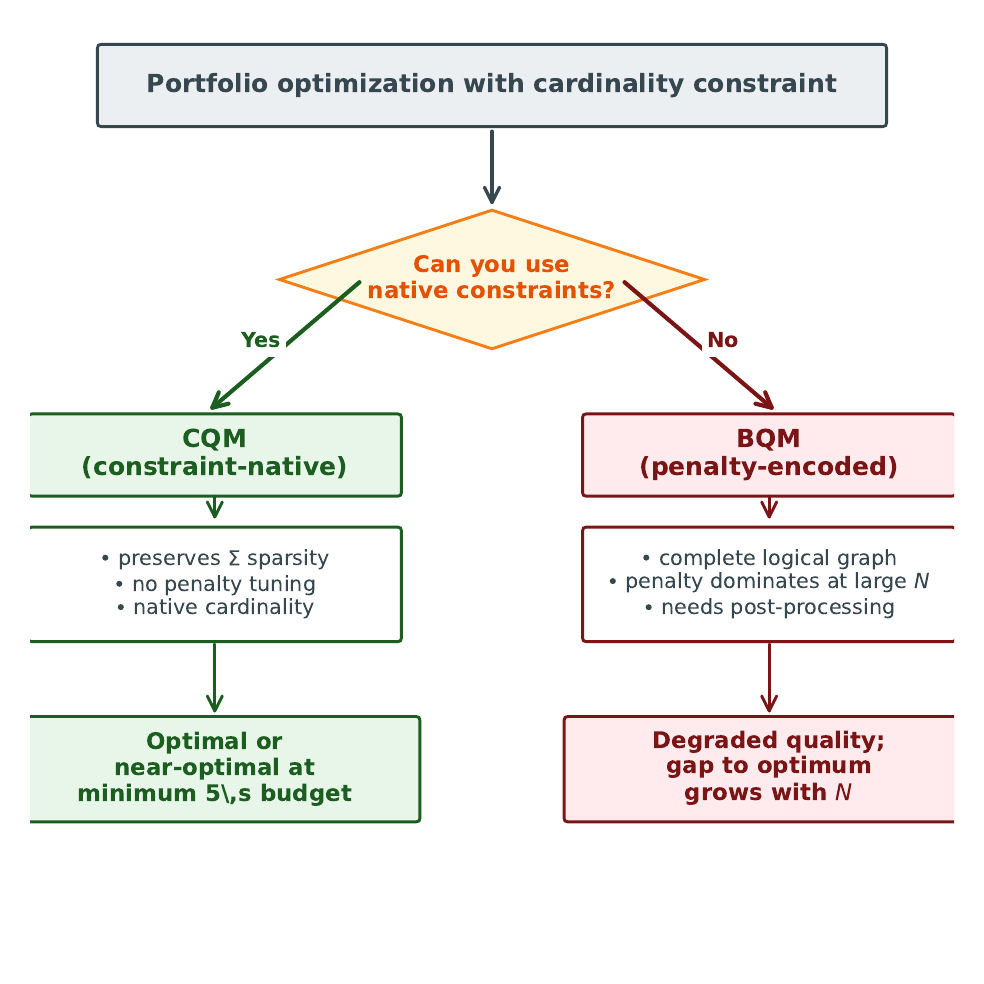}
    \caption{Formulation choice on the tested D-Wave portfolio instances. The
    constraint-native CQM objective preserves sparsity, requires no user-set
    penalty, and matches the available optima at the minimum budget. The
    penalty-encoded BQM creates a dense submitted graph and the hybrid-BQM
    results degrade with $N$ under the tested settings.}
    \label{fig:formulation_choice}
\end{figure}

\subsection{D-Wave Hardware}
\label{sec:hardware}

We use two D-Wave systems:
\begin{itemize}
    \item \textbf{Advantage\_system4.1} (Pegasus topology): 5{,}760 ideal qubits, degree-15
    connectivity. The active qubit count of $5{,}627$ reflects the working graph of
    Advantage\_system4.1 during the experimental campaign (April 2026); embedding
    statistics in Section~\ref{sec:a1} are computed against this active set.
    \item \textbf{Advantage2\_system1.13} (Zephyr topology): 4{,}800 ideal qubits, degree-20
    connectivity.\footnote{On 2026-04-10, D-Wave renamed this solver to
    \texttt{Advantage2\_system1} and removed qubit 4374 (and its couplers) from the working
    graph. Embedding records were generated across this study and the companion
    penalty-encoding paper; we verified that qubit 4374 was not selected by the embedder for
    any chain, so the rename does not affect any reported result. All experiments here were
    run before the rename.}
\end{itemize}
Direct QPU access requires minor embedding of the logical QUBO graph into the hardware
topology. Hybrid solvers (LeapHybridBQMSampler, LeapHybridCQMSampler) accept arbitrary problem
structures and manage embedding and QPU allocation internally, with a minimum time limit of
approximately 5 seconds~\citep{dwave_hybrid}. Figure~\ref{fig:hybrid_blackbox} illustrates the
hybrid solver architecture.

\begin{figure}[t]
    \centering
    \includegraphics[width=0.85\textwidth]{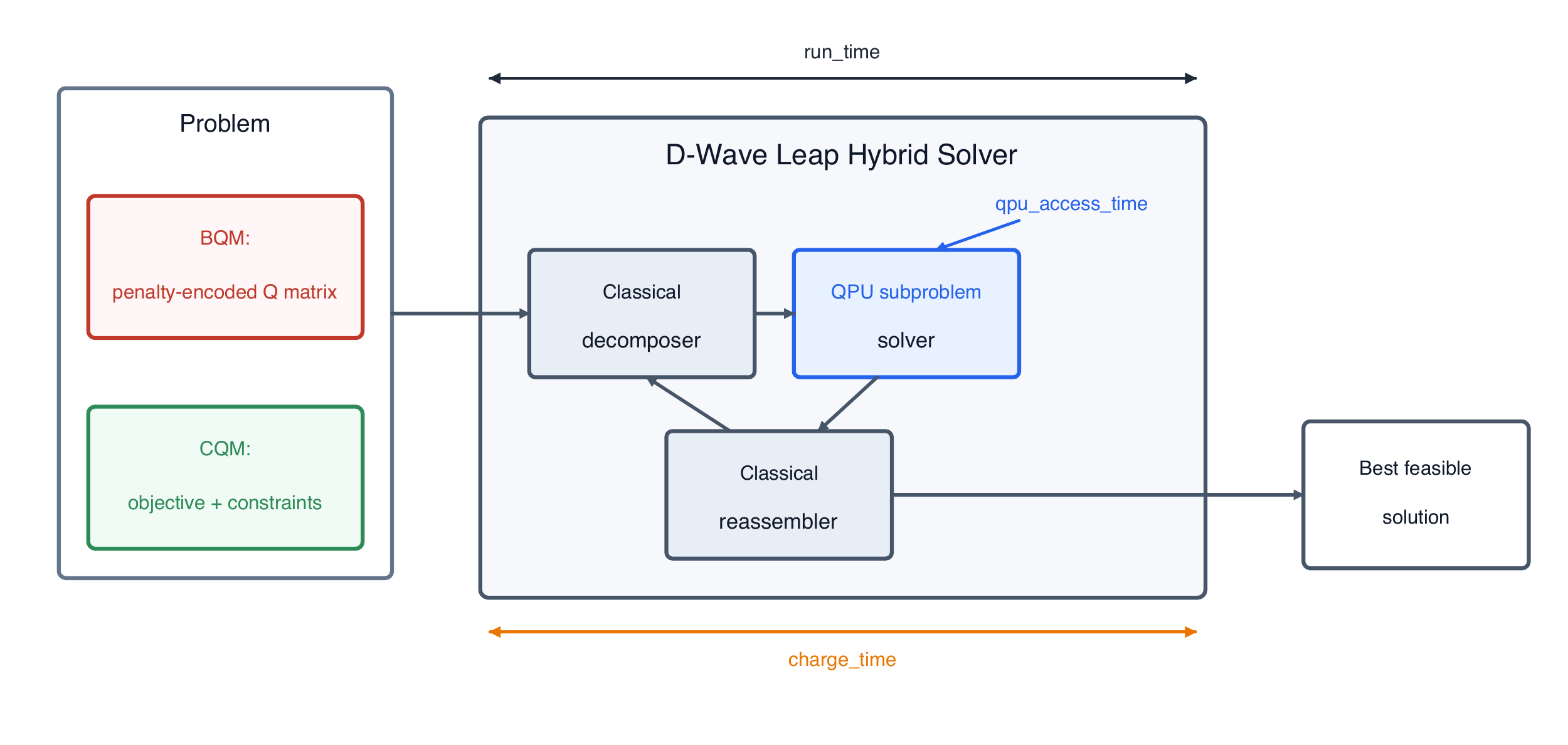}
    \caption{D-Wave hybrid solver architecture. The solver iterates between classical
    decomposition and QPU subproblem solving. Timing fields (\texttt{run\_time},
    \texttt{charge\_time}, \texttt{qpu\_access\_time}) are reported separately, allowing
    analysis of whether the QPU was accessed and for how long; these fields do
    not identify its causal contribution to the final solution.}
    \label{fig:hybrid_blackbox}
\end{figure}

% ============================================================================
\section{Operational Decomposition Audit Protocol}
\label{sec:audit-protocol}
% ============================================================================

The findings of this paper depend on a small set of operational
metrics that are, in principle, observable from any hybrid
quantum-classical solver that exposes per-run resource telemetry.  We
formalise these metrics as a \emph{decomposition audit protocol}.  The
protocol is a reusable methodological contribution independent of the
specific D-Wave LeapHybrid services audited in
Sections~\ref{sec:results}--\ref{sec:discussion}, and is presented as
a \emph{candidate minimum reporting template} for future hybrid-solver
benchmarks, for community discussion.

\subsection{Metrics}
\label{sec:audit-metrics}

Let a hybrid solver run produce a sampleset $\mathcal{S}$ with timing
fields $t_{\mathrm{run}}$ (total run time), $t_{\mathrm{charge}}$
(charged service time), and $t_{\mathrm{QPU}}$ (cumulative QPU-access
time, summed over all sub-problem QPU calls within the run), each as
exposed by the D-Wave Ocean SDK~\citep{dwave_hybrid_solvers_whitepaper};
  a final objective value $f_{\mathcal{S}}$; and a classical reference
  objective $f^\star$ on the same instance (a proven optimum where available,
  otherwise a clearly labelled incumbent).  Let $Q_{\mathrm{encoded}}$ be the
QUBO matrix actually submitted to the solver, $\Sigma$ the original
covariance matrix (or other cost matrix) before encoding, and
$|E(\cdot)|$ the off-diagonal edge count of a matrix (cf.\
$\mathrm{supp_{off}}$ in Theorem~\ref{thm:density_collapse}).  We
define four metrics.

\paragraph{(M1) QPU wall-clock fraction.}
\begin{equation}
  r_{\mathrm{QPU}} \;:=\; \frac{t_{\mathrm{QPU}}}{t_{\mathrm{run}}}
  \;\geq\; 0.
  \label{eq:rQPU}
\end{equation}
$r_{\mathrm{QPU}}$ measures the fraction of wall-clock that the
service spent on the quantum processor, summed over all sub-problem
QPU calls.  It is a resource-accounting quantity, not an indicator of
causal importance: high $r_{\mathrm{QPU}}$ does not
guarantee that the QPU samples were causally important to the final
solution; low $r_{\mathrm{QPU}}$ \emph{does} bound the share of
wall-clock attributable to quantum sampling from above.  QPU-access
time additionally includes programming and readout overhead (across our
$270$ synthetic direct-QPU records, mean programming time is
$25.2$\,ms of a $171.0$\,ms mean access window), so $t_{\mathrm{QPU}}$
\emph{upper-bounds} the pure-anneal share of wall-clock; the fraction
attributable to annealing itself is smaller still.  In D-Wave
Ocean, $t_{\mathrm{QPU}}$ is the per-submission cumulative QPU access
time, read either as
\texttt{sampleset.info['qpu\_access\_time']}
or, when nested timing is exposed, as
\texttt{sampleset.info['timing']['qpu\_access\_time']}.
For the audited Leap records, QPU accesses are reported within the
service run and $r_{\mathrm{QPU}}\leq1$.  On a platform that sums
concurrent QPU calls, the same ratio may exceed one and should be
labelled cumulative QPU-time intensity rather than a literal wall-clock
fraction.

\paragraph{(M2) Optimality gap.}
\begin{equation}
  g_F \;:=\; \frac{f_{\mathcal{S}} \;-\; f^\star}
                  {\max\!\big(1,\, |f^\star|\big)}.
  \label{eq:gF}
\end{equation}
$g_F$ is the regularised objective gap to the classical reference.  The
denominator is regularised to handle $|f^\star| \to 0$.  A negative
$g_F$ indicates improvement over an incumbent; for a valid proven optimum
it should not occur.  $g_F$ is a \emph{solver-side} metric;
portfolio-level financial metrics (e.g., Sharpe and drawdown) are
computed on realised return time series with their own denominators
(Section~\ref{sec:ff49-financial}) and are not interchangeable with $g_F$.

\paragraph{(M3) Encoded-graph density amplification.}
Let $E_{\mathrm{enc}} := |E(Q_{\mathrm{encoded}})|$ and
$E_{\mathrm{orig}} := |E(\Sigma)|$.  Define
\begin{equation}
  d_{\mathrm{amp}} \;:=\;
  \begin{cases}
    E_{\mathrm{enc}} / E_{\mathrm{orig}} & \text{if } E_{\mathrm{orig}} > 0, \\
    \infty                                 & \text{if } E_{\mathrm{orig}} = 0
                                                \text{ and } E_{\mathrm{enc}} > 0, \\
    1                                       & \text{if } E_{\mathrm{enc}} = E_{\mathrm{orig}} = 0.
  \end{cases}
  \label{eq:damp}
\end{equation}
By Theorem~\ref{thm:density_collapse}, penalty-encoded paths satisfy
$d_{\mathrm{amp}} = \binom{N}{2}/E_{\mathrm{orig}}$ when
$E_{\mathrm{orig}} > 0$, and $d_{\mathrm{amp}} = \infty$ when $\Sigma$
is diagonal; constraint-native CQM paths satisfy $d_{\mathrm{amp}} = 1$
by the CQM support-preservation remark after
Theorem~\ref{thm:density_collapse}.  We recommend reporting raw
$E_{\mathrm{enc}}$ and $E_{\mathrm{orig}}$ alongside the ratio to keep
the diagonal-$\Sigma$ case visible.

\paragraph{(M4) Service-output objective variance.}
\begin{equation}
  v_{\mathrm{service}}(N, \rho, B)
  \;:=\;
  \mathrm{Var}\!\Big(f_{\mathcal{S}} \;\Big|\; N,\, \rho,\, B,\,
                        r \in \mathcal{R}\Big),
  \label{eq:vservice}
\end{equation}
the variance of the solver's final objective at fixed problem
dimension $N$, density family $\rho$, and wall-clock budget $B$, over
a set $\mathcal{R}$ of repeated calls with $|\mathcal{R}| \geq 10$.
$v_{\mathrm{service}} = 0$ indicates objective stability in the saved
records on the tested cell; it does not establish identical decision vectors
or deterministic internal execution.  The $|\mathcal{R}| = 1$ case is excluded because
$v_{\mathrm{service}} = 0$ would be tautological.

\subsection{Audit procedure}
\label{sec:audit-procedure}

We propose that any hybrid quantum-classical benchmark report the
following items.

\begin{enumerate}
\item[(P1)] \textbf{All three timing fields} per run: $t_{\mathrm{run}}$,
  $t_{\mathrm{charge}}$, $t_{\mathrm{QPU}}$.  Do not report only
  $t_{\mathrm{QPU}}$.
\item[(P2)] \textbf{$r_{\mathrm{QPU}}$ distribution} per $(N, \rho)$
  cell, with at least mean and inter-quartile range; not a point
  estimate.
\item[(P3)] \textbf{Optimality gap and matched-budget baseline.} Compute
  $g_F$ against the strongest available proven optimum or bound, and report
  separately a strong classical heuristic at matched wall-clock.  Report each
  classical solver's name, version, hardware, stopping rule, and wall-clock;
  do not present a longer-budget optimum as a time-matched comparison.
\item[(P4)] \textbf{Encoded-graph density amplification
  $d_{\mathrm{amp}}$} for every formulation variant.  Report raw
  $E_{\mathrm{enc}}$ and $E_{\mathrm{orig}}$ alongside the ratio.
  Explicitly contrast penalty-encoded and constraint-native paths.
\item[(P5)] \textbf{Feasibility rate before and after post-processing
  repair}.  For penalty-encoded paths, separate the raw-sample
  feasibility rate from the post-projection feasibility rate.
\item[(P6)] \textbf{Service-variance $v_{\mathrm{service}}$} over
  $\geq 10$ repeated calls per cell.  Distinguishes objective-stable
  cells from cells with measurable output variation.
\item[(P7)] \textbf{Solver identity and campaign window}.  For D-Wave
  Leap, record the sampler's solver identity at submission time
  (\texttt{sampler.solver.identity}, or
  \texttt{sampler.solver.name} as a fallback) plus the submission
  \texttt{problem\_id}.  The legacy
  \texttt{sampleset.info['solver']}
  field is unreliable across SDK versions and should not be used as
  the canonical version source.  Also record the start and end dates
  of the run campaign.
\end{enumerate}

\subsection{Self-application and protocol compliance}
\label{sec:audit-self-application}

The protocol is also a checklist against which this campaign can fail.
Table~\ref{tab:protocol_compliance} records that self-audit rather than
retrofitting the requirements to the available files.

\begin{table}[ht]
\centering
\caption{Self-application of the proposed reporting procedure. ``Partial''
marks information that the original campaign did not persist or a grid that
was not fully repeated; no missing field is reconstructed as if observed.}
\label{tab:protocol_compliance}
\begin{tabular}{@{}p{0.06\textwidth}p{0.12\textwidth}p{0.73\textwidth}@{}}
\toprule
Item & Status & Evidence or limitation \\
\midrule
P1 & Partial & All three fields retained for the 162-run main grid; later budget/repeat records omit fields. \\
P2 & Complete & Mean/IQR reported by $N$ in Table~\ref{tab:timing_breakdown} and by family in the data release. \\
P3 & Partial & Gurobi proof/bound plus 5-s Tabu; no proof at $N\geq400$. \\
P4 & Complete$^\dagger$ & Objective-graph counts reconstructed exactly from saved inputs. \\
P5 & Partial & Main BQM/CQM and synthetic-QPU rates retained; equity raw sample absent. \\
P6 & Partial & Ten repeats per planned cell except one cell with four calls. \\
P7 & Partial & Campaign and QPU identities retained; hybrid IDs/problem IDs absent. \\
\bottomrule
\end{tabular}
\\[2pt]
\footnotesize $^\dagger$ Formulation-level support only; proprietary service-internal
transformations are not observable.
\end{table}

Concretely, the main synthetic campaign contains raw-feasible solutions in
$81/81$ hybrid-BQM calls and feasible solutions in $81/81$ hybrid-CQM calls.
In the synthetic direct-QPU campaign, $0/270$ retained best raw
solutions satisfy exact-$K$ before repair and $270/270$ do after
projection.  The equity-QPU
records retained only the projected selection, so their raw feasibility rate
cannot be recovered.  Hybrid and direct-QPU campaigns ran from 7--9 April
2026 (UTC).  The records identify the QPU backends reported in
Section~\ref{sec:hardware}, but do not retain the programmatic hybrid-backend
identity or per-submission problem ID.  These omissions motivate P5--P7 and
limit cross-version replication of the service-level findings.
The main 162-run hybrid grid retains all three timing fields; the
78-run budget sweep omits charge time, while the 228-run repeated-call
file omits both run and charge time.  Consequently, timing-distribution
claims use the main grid, and the later files support only the fields
they actually retain.

\subsection{Applicability}
\label{sec:audit-applicability}

The protocol applies directly to hybrid services that expose the three
timing fields above as first-class telemetry: specifically, D-Wave
Leap hybrid solvers (LeapHybridCQMSampler and LeapHybridBQMSampler),
which return \texttt{run\_time}, \texttt{charge\_time}, and
\texttt{qpu\_access\_time} on every submission.  Adapting the protocol
to other hybrid services (AWS Braket Hybrid Jobs, Qiskit Runtime
Estimator-based hybrid optimisers) is possible but requires
platform-specific instrumentation: those services expose only partial
or aggregate metrics (e.g., quantum-task counts and latencies, or
\texttt{quantum\_seconds}) rather than a direct
$(t_{\mathrm{run}}, t_{\mathrm{charge}}, t_{\mathrm{QPU}})$ tuple.
We do not validate the protocol on these alternative services in the
present paper.

\subsection{Reference implementation}
\label{sec:audit-reference-impl}

A reference implementation of the four metrics, ingesting a saved
hybrid sampleset record and emitting a structured JSON audit report,
accompanies this paper as
\path{analysis/operational_audit.py}, with metric functions in
\path{paper2.analysis.audit_protocol}.  The same repository contains
the analysis scripts that produced every empirical result reported in
Sections~\ref{sec:results}--\ref{sec:discussion}, including the
matched-budget classical-reproducibility comparison
(Section~\ref{sec:qpu-replacement-ablation}).

\subsection{Interpretive limits}
\label{sec:audit-limits}

The protocol has several interpretive limits, which we state
explicitly.

\paragraph{$r_{\mathrm{QPU}}$ is wall-clock, not causal.}
A run with $r_{\mathrm{QPU}} = 0.5$ may still be classical-dominated
if the QPU samples are ignored by the reassembler; a run with
$r_{\mathrm{QPU}} = 0.007$ may still have a QPU contribution that is
structurally necessary if the classical decomposer's logic is tuned to
QPU-shaped sub-problems.  The protocol does not resolve this
ambiguity.  Resolving it requires a \emph{QPU-replacement ablation} in
which the QPU sub-solver is replaced by an equivalent classical
sub-solver at matched per-call budget, \emph{inside} the pipeline.  We
cannot perform that experiment on a proprietary service.
Section~\ref{sec:qpu-replacement-ablation} instead provides an
external, matched-budget classical-reproducibility comparison, which
tests objective achievability but not internal causality.

\paragraph{Solver-version drift.}
Commercial hybrid services version their internal solver behind a
stable user-facing interface; timing distributions and qualitative
behaviour can change between solver-version updates.  The protocol
therefore requires (P7) the explicit recording of the solver identity
and campaign window, and benchmark results should not be extrapolated
across version boundaries without re-running the audit.  We flag a
  self-audit finding: the campaigns reported here persisted neither
  per-run problem IDs nor the programmatic hybrid solver-identity string in
  the released records.  QPU backend identities and the campaign window are
  recoverable, but the hybrid backend version is not.  Full P7 compliance
  requires capturing all identifiers at run time.

\paragraph{Working-graph and calibration drift.}
For direct-QPU paths, the working graph (active qubits and couplers)
and per-qubit calibration can drift across calibration cycles.  Runs
within a single calibration cycle are comparable; runs across cycles
are not, unless the calibration is explicitly recorded and accounted
for.

\paragraph{Timing-field granularity and client/service wall-clock.}
$t_{\mathrm{run}}$ as reported by the service is the service-side
wall-clock; the client-side wall-clock (including network round-trips
and queue waits) is generally larger and depends on operational
conditions.  $r_{\mathrm{QPU}}$ uses the service-side denominator and
therefore overstates the QPU's share of the user-observed latency.

\paragraph{No latency or queue-time metrics.}
In production deployment, queue depth, batch latency, and
service-level objectives matter as much as the four metrics above.
The protocol is a measurement framework for solver behaviour, not a
service-level-agreement framework.

% ============================================================================
\section{Experimental Setup}
\label{sec:setup}
% ============================================================================

\subsection{Synthetic Controlled Instances}

We generate controlled benchmark instances across three covariance-density families:
\begin{itemize}
    \item \textbf{Diagonal}: near-banded covariance with light nearest-neighbour off-diagonal
    noise ($\epsilon = 0.02$); per-instance density decays as $\rho \approx 2/N$, from
    $\approx 0.20$ at $N{=}10$ to below $0.01$ at $N \geq 200$.
    \item \textbf{Block-diagonal}: equicorrelated blocks with zero cross-block interaction;
    per-instance density $\rho \in [0.11, 0.25]$ depending on $N$ and block structure.
    \item \textbf{Dense}: full covariance matrices from Wishart draws ($\rho \approx 1.0$).
\end{itemize}
The structural results in Section~\ref{sec:density_collapse} depend only on the support
pattern of the encoded QUBO and are invariant to any positive linear rescaling of
$\boldsymbol{\mu}$. The empirical CQM-vs-BQM quality comparison in
Section~\ref{sec:a2a3} is, by contrast, scale-dependent: at the chosen parameterization the
penalty term $A \cdot \mathbf{1}\mathbf{1}^\top$ with $A=4.0$ dominates the linear return
contribution at large $N$.  The resulting coefficient-scale and density burden coincides
with the observed LeapHybridBQM degradation under the tested pipeline, while the Tabu
comparison shows that it does not preclude high-quality BQM solutions. Synthetic instances use
$\boldsymbol{\mu} \sim \mathcal{N}(0.05, 0.02)$ as unitless signal coefficients scaled for
numerical stability; per-variable variance scales are drawn independently from
    $\mathcal{U}(0.02, 0.08)$. Financial metrics (Sharpe ratio and drawdown) are
reported only on the realized Fama--French 49 industry portfolio data in
Section~\ref{sec:a4}, not on the synthetic instances. All covariance matrices are symmetric
positive definite by construction.  The main size sweep uses three
generator seeds per $(N,\text{family})$ cell.  The budget sweep and
repeated-call study instead hold the generated instance fixed at seed~0
to isolate service-budget and service-output variation.  Cardinality is
set to $K = \lfloor 0.3N \rceil$ and risk aversion to $\lambda = 0.5$.
Every reported experiment passes no previous portfolio and sets
$\boldsymbol{\tau}=\mathbf{0}$, so the evaluated objective is the
single-period MV special case of~\eqref{eq:mvt}.

\subsection{Real Equity Data}

We use the Fama--French 49 industry daily return portfolios~\citep{french_data}, constructing
rolling estimation windows of 252 trading days with monthly rebalancing. Five evenly spaced
windows are selected. Industry subsets of $N \in \{10, 20, 30\}$ are formed by selecting the
industries with the largest absolute expected-return magnitude.

\subsection{Solver Configurations}

\paragraph{Direct QPU.} Forward annealing on both Pegasus and Zephyr with 1000 reads, annealing
time $20\,\mu\text{s}$, chain strengths $\{0.5, 1.0, 2.0\}$, and multi-seed embedding search
(5 seeds, 20 tries per seed). All direct-QPU experiments use the default gauge setting
(single-gauge sampling with no spin-reversal transforms); multi-gauge robustness checks are
deferred to future work.  They also use Ocean's default
\texttt{auto\_scale=True}; chain-strength ratios in
Section~\ref{sec:a1} are therefore stated before the common hardware
rescaling. Post-processing projects each sample to exact-$K$ feasibility
using the greedy projection of~\citet{lozano2026penalty}.

\paragraph{Hybrid BQM.} LeapHybridBQMSampler with the penalty-encoded QUBO~\eqref{eq:qubo},
$A = 4.0$. The penalty value is consistent with the calibration study
of~\citet{lozano2026penalty}; we verify robustness to $A \in \{2.0, 4.0, 8.0\}$ in
Appendix~\ref{app:penalty}. Post-processing projects to exact-$K$.

\paragraph{Hybrid CQM.} LeapHybridCQMSampler with the constraint-native
formulation~\eqref{eq:cqm}. Cardinality is enforced natively; the solver returns feasible
solutions directly.

Both hybrid interfaces were called through \texttt{dwave-system} 1.34.0
between 7 and 9 April 2026 (UTC).  As disclosed in
Section~\ref{sec:audit-self-application}, the saved records retain sampler
class and timing fields but not the service's programmatic backend-version
string or per-submission problem ID.

\paragraph{Gurobi MIQP.} The $\boldsymbol{\tau}=\mathbf{0}$ special case
of~\eqref{eq:mvt} is solved as a native mixed-integer quadratic program
using Gurobi 13.0 (Gurobi Optimization, LLC) with a
300-second time limit. Gurobi proves optimality on $62$ of $63$ attempted
instances at $N \leq 200$; one $N{=}200$ run reaches the time limit.  At
$N \geq 400$ the
academic license does not support the model size. Gurobi serves as the gold-standard
reference where optimality is proven.

\paragraph{Simulated annealing.} D-Wave's \texttt{neal} simulated-annealing sampler (1000
reads, 1000 sweeps) applied to the penalty-encoded BQM~\eqref{eq:qubo}. Solutions are projected
to exact-$K$ with the same post-processing as the QPU path.

\subsection{Benchmark Axes}

\begin{enumerate}
    \item \textbf{Problem size $N$}: $\{10, 20, 30, 50, 80\}$ for direct QPU;
    $\{10, 20, 30, 50, 80, 120, 200, 400, 640\}$ for hybrid and classical solvers.
    \item \textbf{Covariance density}: diagonal, block, dense.
    \item \textbf{Wall-clock budget}: $\{5, 30, 60, 180, 300\}$ seconds for hybrid budget
    sweep (Section~\ref{sec:b2}).
\end{enumerate}

% ============================================================================
\section{Results}
\label{sec:results}
% ============================================================================

\subsection{Direct QPU: Chain-Break Fractions and Embedding Overhead}
\label{sec:a1}

Table~\ref{tab:a1_cbf} reports the mean chain-break fraction and embedding overhead for direct
QPU sampling of the penalty-encoded QUBO, and Figure~\ref{fig:chainbreak} visualizes these
trends split by hardware topology (Pegasus vs Zephyr). Chain-break fractions increase
monotonically with $N$, reaching $0.937 \pm 0.016$ at $N = 80$. The embedding overhead grows
from $1.6\times$ at $N = 10$ to $9.4\times$ at $N = 80$. These trends are consistent with the
cross-generation chain-break and embedding-overhead measurements reported by~\citet{pelofske2025}
for combinatorial benchmark problems on Pegasus and Zephyr.

\begin{table}[ht]
\centering
\caption{Direct QPU embedding statistics (mean $\pm$ std pooled over all families, seeds,
chain strengths, \emph{and both hardware topologies}; 270 runs, all embedded successfully).
Per-topology splits appear in Figure~\ref{fig:chainbreak}; the $N{=}80$ per-topology
physical-qubit means are $769$ (Pegasus) and $731$ (Zephyr).  Chain-break fractions and embedding
overheads are nearly identical across density families at each $N$, confirming the density-axis
collapse of Section~\ref{sec:density_collapse}.}
\label{tab:a1_cbf}
\begin{tabular}{rrrrr}
\toprule
$N$ & $K$ & Chain-break frac. & Phys.\ qubits & Overhead \\
\midrule
10 & 3 & $0.543 \pm 0.058$ & 16 & $1.6\times$ \\
20 & 6 & $0.756 \pm 0.034$ & 51 & $2.5\times$ \\
30 & 9 & $0.820 \pm 0.019$ & 107 & $3.6\times$ \\
50 & 15 & $0.894 \pm 0.023$ & 298 & $6.0\times$ \\
80 & 24 & $0.937 \pm 0.016$ & 750 & $9.4\times$ \\
\bottomrule
\end{tabular}
\end{table}

\begin{figure}[ht]
\centering
\includegraphics[width=\textwidth]{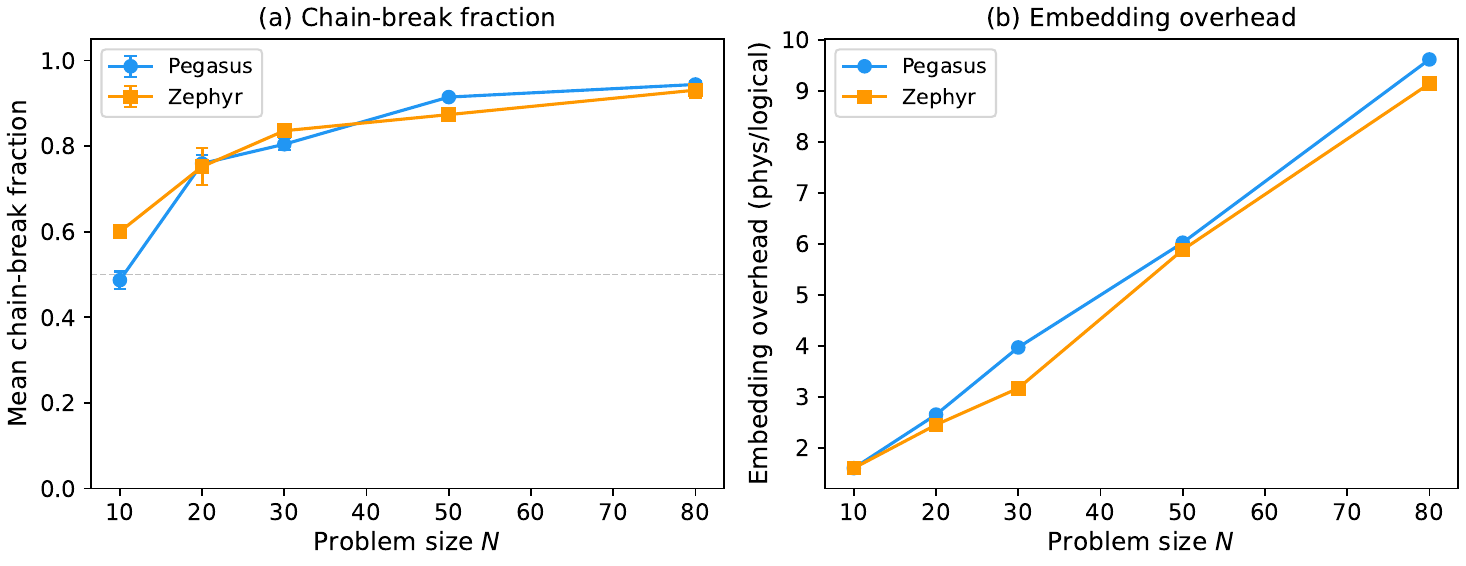}
\caption{(a) Mean chain-break fraction and (b) embedding overhead (physical / logical qubits)
vs problem size for direct QPU on penalty-encoded QUBOs, split by hardware topology.}
\label{fig:chainbreak}
\end{figure}

At mean chain-break fractions above 0.9, more than $90\%$ of logical chains
are broken per returned sample on average.  The decoded outputs therefore
depend heavily on chain unembedding and exact-$K$ projection.  We read this
subsection as characterizing embedding feasibility and
projector behaviour under the tested protocol (single-gauge sampling, a single un-swept
$20\,\mu$s anneal), not as a characterization of annealing dynamics.  The embedding-overhead
growth is a direct consequence of the $A \cdot \mathbf{1}\mathbf{1}^\top$ penalty term,
which forces the complete-graph embedding regardless of the original $\Sigma$ density.
Figure~\ref{fig:density_data} confirms this empirically: chain-break fractions are nearly
identical across diagonal, block, and dense families at each $N$, on both Pegasus and Zephyr.

\begin{figure}[t]
    \centering
    \includegraphics[width=\textwidth]{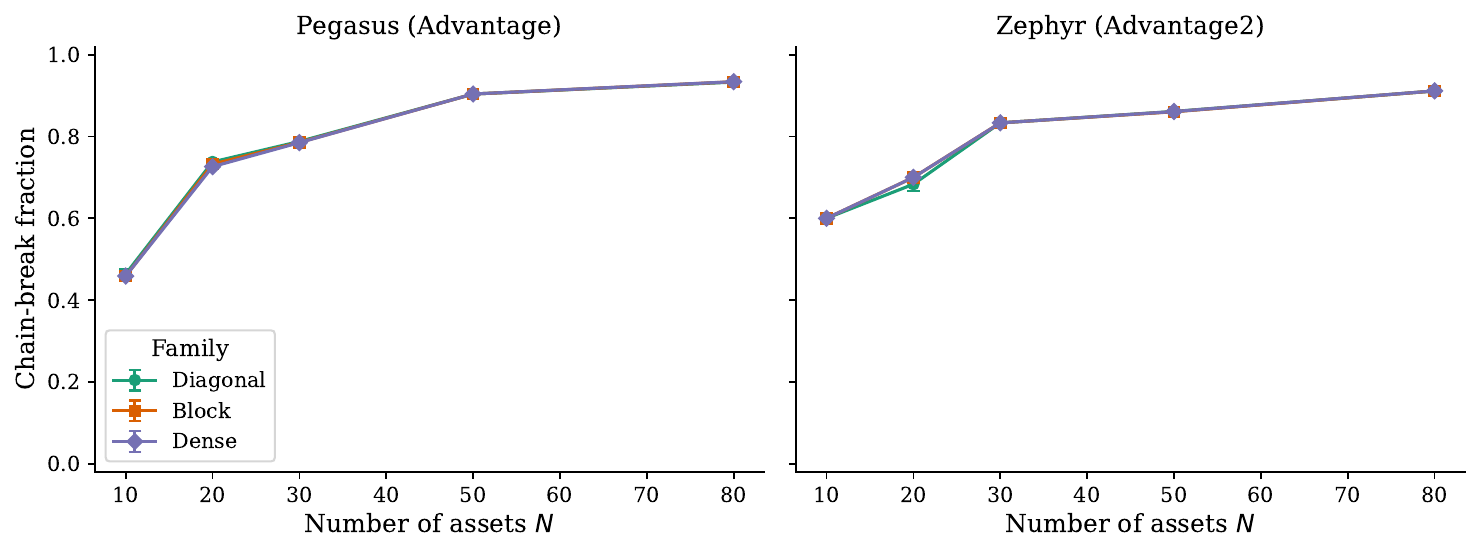}
    \caption{Chain-break fraction vs $N$ split by covariance-density family (diagonal, block,
    dense), on Pegasus (left) and Zephyr (right). The three curves nearly overlap at each $N$,
    confirming that penalty encoding collapses the density axis.}
    \label{fig:density_data}
\end{figure}

\paragraph{Per-chain-strength refinement.}
Table~\ref{tab:a1_cbf} averages over chain strengths
$J_c \in \{0.5, 1.0, 2.0\}$.  Under the symmetric-QUBO convention used
throughout ($H = z^\top Q z$ with $Q$ symmetric), the cardinality
penalty contributes an upper-triangular QUBO coefficient $2A=8$ per
unordered pair.  Ocean's embedding composite converts the BQM to SPIN
before adding chains, where that penalty contribution maps to the pairwise
Ising coupling $A/2=2$.  The tested strengths
$J_c\in\{0.5,1.0,2.0\}$ therefore span $25\%$--$100\%$ of the
penalty-only pairwise Ising scale before the common hardware rescaling.
That ratio alone does not establish adequate chain strength on a submitted
$K_N$ graph of degree $N-1$, and the experiment does not include Ocean's
torque-compensated default or a larger scale-aware sweep.
Table~\ref{tab:cbf_per_strength} splits the same data by chain
strength.  Within the tested regime: at the strongest tested
$J_c = 2.0$, CBF is already $\approx 0.81$ at $N = 30$, exceeds
$0.88$ by $N = 50$, and reaches $0.92$ at $N = 80$, and the CBF
decrease from $J_c = 0.5$ to $2.0$ is monotone but small
($\sim 0.03$ at $N = 80$).  Because the sweep omits torque-compensated
and larger scale-aware values, these data characterize the specified grid
only and cannot separate a structural chain-break floor from chain-strength
tuning; a broader scale-aware sweep is future work.  The structural
claim we retain is Corollary~\ref{cor:embedding_lower_bound}'s: the
embedding-overhead growth $\Omega(N^2/d)$, and with it the
chain-length growth of Lemma~\ref{lem:embedding_growth}, is
independent of chain strength.

\begin{table}[!ht]
\centering
\caption{Direct-QPU chain-break fraction split by chain strength $J_c$
(mean $\pm$ sample std over density families, seeds, and topologies).
The penalty contributes an upper-triangular QUBO coefficient $2A=8$
per pair; Ocean converts the BQM to SPIN before embedding, where the
penalty contribution maps to an Ising coupling $A/2=2$.  Thus
$J_c \in \{0.5,1.0,2.0\}$ spans $25\%$--$100\%$ of that pairwise
scale.  This ratio alone does not establish adequate chain strength on
the complete graph, and no torque-compensated or larger scale-aware
sweep was run; the table is a diagnostic split, not a tuned bound.
This per-strength split refines the chain-strength-averaged summary in
Table~\ref{tab:a1_cbf}.}
\label{tab:cbf_per_strength}
\small
\begin{tabular}{rrrrr}
\toprule
$N$ & $J_c$ & runs & CBF mean $\pm$ std & Phys.\ overhead \\
\midrule
10 & $0.5$ & 18 & $0.550 \pm 0.051$ & $1.6\times$ \\
10 & $1.0$ & 18 & $0.550 \pm 0.051$ & $1.6\times$ \\
10 & $2.0$ & 18 & $0.530 \pm 0.073$ & $1.6\times$ \\
\addlinespace
20 & $0.5$ & 18 & $0.785 \pm 0.014$ & $2.5\times$ \\
20 & $1.0$ & 18 & $0.768 \pm 0.006$ & $2.5\times$ \\
20 & $2.0$ & 18 & $0.713 \pm 0.023$ & $2.5\times$ \\
\addlinespace
30 & $0.5$ & 18 & $0.828 \pm 0.012$ & $3.6\times$ \\
30 & $1.0$ & 18 & $0.822 \pm 0.013$ & $3.6\times$ \\
30 & $2.0$ & 18 & $0.810 \pm 0.024$ & $3.6\times$ \\
\addlinespace
50 & $0.5$ & 18 & $0.905 \pm 0.018$ & $6.0\times$ \\
50 & $1.0$ & 18 & $0.895 \pm 0.024$ & $6.0\times$ \\
50 & $2.0$ & 18 & $0.882 \pm 0.022$ & $6.0\times$ \\
\addlinespace
80 & $0.5$ & 18 & $0.955 \pm 0.002$ & $9.4\times$ \\
80 & $1.0$ & 18 & $0.934 \pm 0.011$ & $9.4\times$ \\
80 & $2.0$ & 18 & $0.922 \pm 0.011$ & $9.4\times$ \\
\bottomrule
\end{tabular}
\end{table}

\FloatBarrier
\subsection{Hybrid CQM vs BQM: Formulation Comparison}
\label{sec:a2a3}

Table~\ref{tab:a2a3} compares the mean objective values achieved by hybrid BQM and CQM across
problem sizes. All values are evaluated on the original constrained
objective~\eqref{eq:mvt} with $\boldsymbol{\tau}=\mathbf{0}$: BQM
solutions are projected to exact-$K$ before evaluation.

\begin{table}[ht]
\centering
\caption{Hybrid CQM vs BQM: mean objective value by $N$ (averaged over 3 density families
$\times$ 3 seeds = 9 instances per $N$; 162 runs total). Lower is better. ``CQM strict wins''
counts instances with $f_{\mathrm{CQM}} < f_{\mathrm{BQM}} - \varepsilon$,
$\varepsilon = 10^{-6}$; ``ties'' counts $|f_{\mathrm{CQM}} - f_{\mathrm{BQM}}| \leq \varepsilon$.
A paired one-sided Wilcoxon signed-rank test on per-instance objectives rejects
$H_0: f_{\mathrm{CQM}} \geq f_{\mathrm{BQM}}$ at every $N \geq 20$ after
Holm correction across the eight non-degenerate size-wise tests
($p_{\mathrm{Holm}} = 0.031$ at $N = 20$ and $0.016$ at $N \geq 30$);
at $N{=}10$ all nine pairs are ties.}
\label{tab:a2a3}
\begin{tabular}{rrrrr}
\toprule
$N$ & BQM objective & CQM objective & CQM strict wins & Ties \\
\midrule
10 & $-0.139$ & $-0.139$ & 0/9 & 9 \\
20 & $-0.310$ & $-0.315$ & 5/9 & 4 \\
30 & $-0.417$ & $-0.455$ & 9/9 & 0 \\
50 & $-0.540$ & $-0.654$ & 9/9 & 0 \\
80 & $-0.622$ & $-0.955$ & 9/9 & 0 \\
120 & $-0.628$ & $-1.247$ & 9/9 & 0 \\
200 & $+0.024$ & $-1.525$ & 9/9 & 0 \\
400 & $+4.480$ & $+0.040$ & 9/9 & 0 \\
640 & $+16.452$ & $+6.229$ & 9/9 & 0 \\
\bottomrule
\end{tabular}
\end{table}

At $N \geq 200$, the mean hybrid-BQM objective becomes positive.  This
deterioration is consistent with the encoded penalty scale and complete-graph
support burdening the hybrid-BQM search, but the opaque telemetry does not
reveal how its budget is allocated internally. Under the exact-$K$ cardinality constraint, a positive objective value indicates a
portfolio whose risk-adjusted cost exceeds its expected return. The CQM path,
which handles cardinality natively, does not exhibit this failure mode under
the tested service settings.

Figure~\ref{fig:objective_gap} visualizes the solver comparison including classical baselines.

\begin{figure}[ht]
\centering
\includegraphics[width=0.85\textwidth]{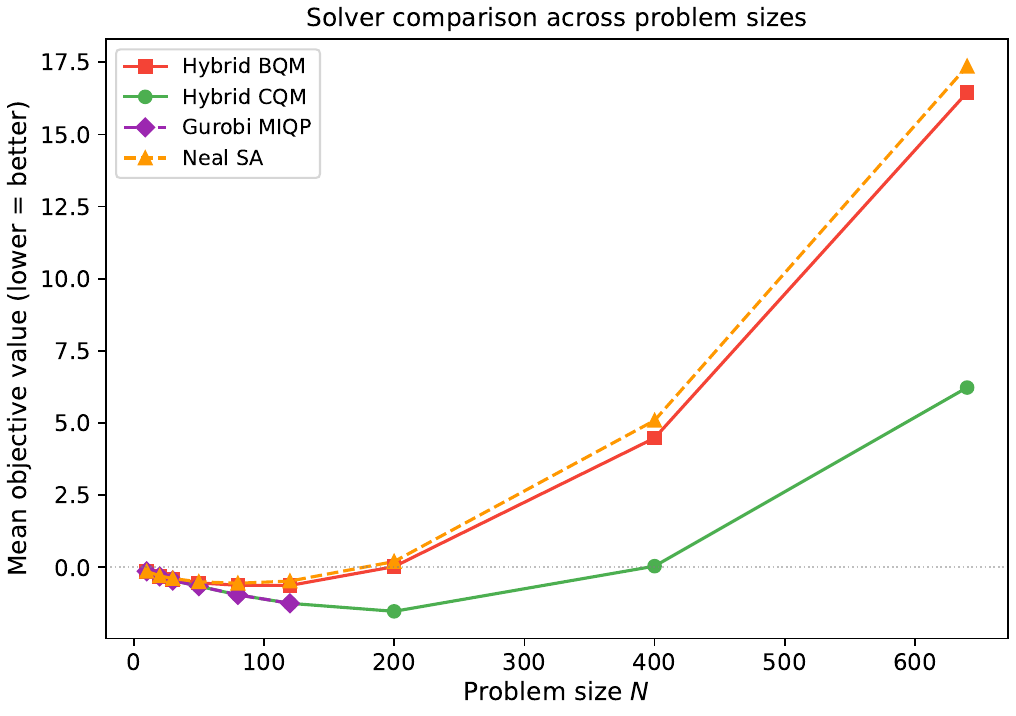}
\caption{Mean objective value vs problem size across solver families (averaged over 3 density
families $\times$ 3 seeds). Gurobi MIQP (dashed purple) is the classical
reference, with optimality proved for $62/63$ attempted instances at
$N\leq200$; the single unproved $N=200$ run is not used in the gap table.
Hybrid CQM (green) tracks Gurobi closely. Hybrid BQM (red) and neal SA (orange) both operate
on the penalty-encoded QUBO and diverge at large $N$.}
\label{fig:objective_gap}
\end{figure}

\FloatBarrier
\subsection{Gap to Gurobi Optimal}
\label{sec:gap}

Table~\ref{tab:gap} and Figure~\ref{fig:gap} report the regularized
gap $g_F$ of~\eqref{eq:gF} to Gurobi's proven optimum. CQM matches the
proven optimum exactly (gap $= 0$) on all $54$ head-to-head
instances at $N \leq 120$ (6 sizes $\times$ 9 instances per $N$, all of which Gurobi
proves to optimality). BQM and SA, both operating on the penalty-encoded formulation,
show monotonically increasing gaps.

\begin{table}[ht]
\centering
\caption{Regularized gap $g_F$ to Gurobi optimal (mean $\pm$ sample
std, $\mathrm{ddof}{=}1$, over $9$ instances per $N$), using the
$\max(1,|f^\star|)$ denominator of~\eqref{eq:gF}. CQM achieves zero
gap at every reported $N$. BQM and SA gaps grow with $N$ under the
tested settings.}
\label{tab:gap}
\begin{tabular}{rrrr}
\toprule
$N$ & CQM gap & BQM gap & SA gap \\
\midrule
10 & $0.000$ & $0.000$ & $0.000$ \\
20 & $0.000$ & $0.005 \pm 0.007$ & $0.019 \pm 0.012$ \\
30 & $0.000$ & $0.038 \pm 0.013$ & $0.066 \pm 0.014$ \\
50 & $0.000$ & $0.115 \pm 0.018$ & $0.154 \pm 0.030$ \\
80 & $0.000$ & $0.288 \pm 0.074$ & $0.350 \pm 0.106$ \\
120 & $0.000$ & $0.455 \pm 0.278$ & $0.571 \pm 0.355$ \\
\bottomrule
\end{tabular}
\end{table}

\begin{figure}[ht]
\centering
\includegraphics[width=0.75\textwidth]{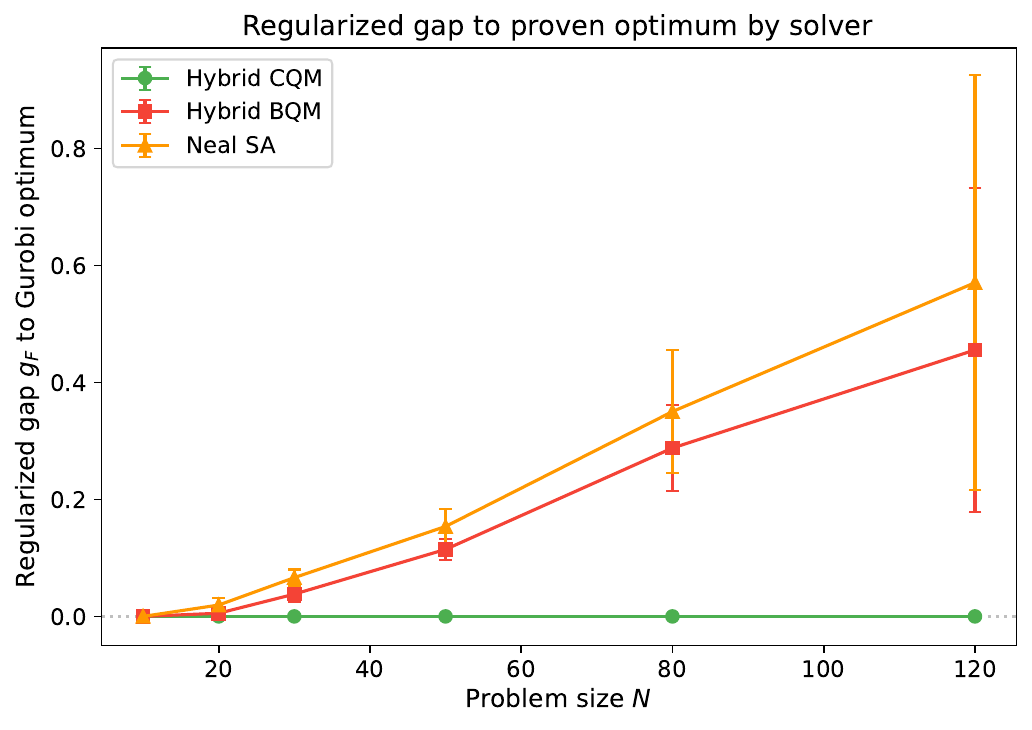}
\caption{Regularized gap $g_F$ to Gurobi optimal vs $N$. CQM (green)
is indistinguishable from the optimal baseline; BQM (red) and SA
(orange) diverge under the tested penalty-encoded settings.}
\label{fig:gap}
\end{figure}

\FloatBarrier
\subsection{Classical Baselines: Gurobi MIQP and Simulated Annealing}
\label{sec:classical}

Gurobi MIQP proves optimality for $62$ of $63$ attempted instances at
$N \leq 200$; one $N = 200$ run reaches the time limit. At $N \geq 400$,
the free academic license does not
support the model size. The neal simulated-annealing heuristic, which operates on the same
penalty-encoded BQM as the hybrid BQM solver, exhibits a similar degradation pattern: objective
quality deteriorates at large $N$ under its tested settings.  This shows that
penalty encoding can burden more than one search procedure, but the Tabu
comparison in Section~\ref{sec:qpu-replacement-ablation} also shows that the
observed hybrid-BQM degradation is not an unavoidable consequence of the
formulation alone.

\FloatBarrier
\subsection{Full Hybrid Timing Breakdown}
\label{sec:timing-breakdown}

Table~\ref{tab:timing_breakdown} reports the three timing fields the D-Wave
hybrid SDK exposes per submission ($t_{\mathrm{run}}$,
$t_{\mathrm{charge}}$, and $t_{\mathrm{QPU}}$) as means over the 9
instances per $(N, \mathrm{solver})$ cell of the synthetic-hybrid
campaign~(162 runs total). The QPU wall-clock fraction
$r_{\mathrm{QPU}} = t_{\mathrm{QPU}}/t_{\mathrm{run}}$
(Section~\ref{sec:audit-metrics}) is included as the rightmost column.

Three observations:
\begin{itemize}
  \item For LeapHybridCQM, $t_{\mathrm{QPU}}$ is approximately constant in
  $N$ ($\sim 34$\,ms across all sizes from $N{=}10$ to $N{=}640$), while
  $t_{\mathrm{run}}$ and $t_{\mathrm{charge}}$ stay close to the
  service's $5$\,s minimum. The result is
  $r_{\mathrm{QPU}} \in [0.67\%, 0.77\%]$ across the entire $N$ range, with
  no monotonic trend; this is the $\approx 0.7\%$ headline number.
  \item For LeapHybridBQM, $t_{\mathrm{QPU}}$ is larger
  ($\sim 147$--$187$\,ms) and varies more across $N$, peaking near $N=200$;
  $r_{\mathrm{QPU}}$ is in the range $[2.95\%, 3.75\%]$, roughly $4\times$
  the CQM value.
  \item In both solvers, $t_{\mathrm{run}} \approx t_{\mathrm{charge}}$
  (within $\sim 3\%$), so the choice of denominator does not materially
  change the reported QPU fraction.
\end{itemize}

The fraction of runs with recorded QPU access (\texttt{qpu\_access\_time} $> 0$) is
$80/81$ for hybrid BQM and $79/81$ for hybrid CQM; in the budget sweep
(Section~\ref{sec:b2}), $74/78$ runs across both solvers record QPU access.

The constancy of $t_{\mathrm{QPU}}$ in $N$ on the CQM path is compatible
with decomposition into bounded-size QPU subproblems, but the unlabelled
residual does not identify where the scale-up cost is paid. The interpretive limits of
Section~\ref{sec:audit-limits} apply: $r_{\mathrm{QPU}}$ bounds the
\emph{wall-clock} share, not the causal contribution.
Sections~\ref{sec:dwell-quality} and
\ref{sec:qpu-replacement-ablation} probe the causal question
with two complementary non-causal diagnostics: dwell--quality
association and matched-budget classical reproducibility.

\begin{table}[ht]
\centering
\caption{Full timing breakdown per $(N,\,\text{solver})$ cell on the synthetic-hybrid runs (mean over 9 instances per cell: 3 density families $\times$ 3 seeds; 162 runs total).  All timing fields in seconds.  $r_{\mathrm{QPU}}$ is shown as the within-cell mean [IQR] for $t_{\mathrm{QPU}}/t_{\mathrm{run}}$.}
\label{tab:timing_breakdown}
\begin{tabular}{rlrrrr}
\toprule
$N$ & Solver & $t_{\mathrm{run}}$ (s) & $t_{\mathrm{charge}}$ (s) & $t_{\mathrm{QPU}}$ (s) & $r_{\mathrm{QPU}}$ mean [IQR] (\%) \\
\midrule
10 & hybrid-bqm & $4.987$ & $4.987$ & $0.1469$ & $2.95\,[0.00]$ \\
20 & hybrid-bqm & $4.987$ & $4.987$ & $0.1516$ & $3.04\,[0.01]$ \\
30 & hybrid-bqm & $4.991$ & $4.991$ & $0.1564$ & $3.13\,[0.01]$ \\
50 & hybrid-bqm & $4.989$ & $4.989$ & $0.1626$ & $3.26\,[0.01]$ \\
80 & hybrid-bqm & $4.994$ & $4.994$ & $0.1855$ & $3.71\,[0.01]$ \\
120 & hybrid-bqm & $4.996$ & $4.996$ & $0.1671$ & $3.34\,[0.97]$ \\
200 & hybrid-bqm & $4.991$ & $4.991$ & $0.1873$ & $3.75\,[1.03]$ \\
400 & hybrid-bqm & $4.994$ & $4.994$ & $0.1553$ & $3.11\,[0.01]$ \\
640 & hybrid-bqm & $4.987$ & $4.987$ & $0.1556$ & $3.12\,[0.00]$ \\
10 & hybrid-cqm & $4.613$ & $4.569$ & $0.0347$ & $0.76\,[0.01]$ \\
20 & hybrid-cqm & $4.514$ & $4.488$ & $0.0347$ & $0.77\,[0.00]$ \\
30 & hybrid-cqm & $4.842$ & $4.745$ & $0.0347$ & $0.72\,[0.11]$ \\
50 & hybrid-cqm & $4.582$ & $4.564$ & $0.0309$ & $0.68\,[0.01]$ \\
80 & hybrid-cqm & $4.683$ & $4.641$ & $0.0309$ & $0.67\,[0.11]$ \\
120 & hybrid-cqm & $4.506$ & $4.488$ & $0.0347$ & $0.77\,[0.00]$ \\
200 & hybrid-cqm & $4.927$ & $4.813$ & $0.0347$ & $0.71\,[0.11]$ \\
400 & hybrid-cqm & $4.937$ & $4.822$ & $0.0347$ & $0.71\,[0.09]$ \\
640 & hybrid-cqm & $4.973$ & $4.826$ & $0.0348$ & $0.70\,[0.10]$ \\
\bottomrule
\end{tabular}
\end{table}

\FloatBarrier
\subsection{Direct QPU on Real Equity Data}
\label{sec:a4}

We confirm the synthetic findings on Fama--French 49 industry data as a directional sanity
check (90 runs). Chain-break fractions follow the same pattern as the synthetic instances:
$0.483$ at $N = 10$, $0.759$ at $N = 20$, and $0.802$ at $N = 30$ (Pegasus means). All 90
embeddings succeeded; the constraint is solution quality, not embedding feasibility.
Figure~\ref{fig:equity_sanity} shows the chain-break trend on real data.

\begin{figure}[t]
\centering
\includegraphics[width=0.75\textwidth]{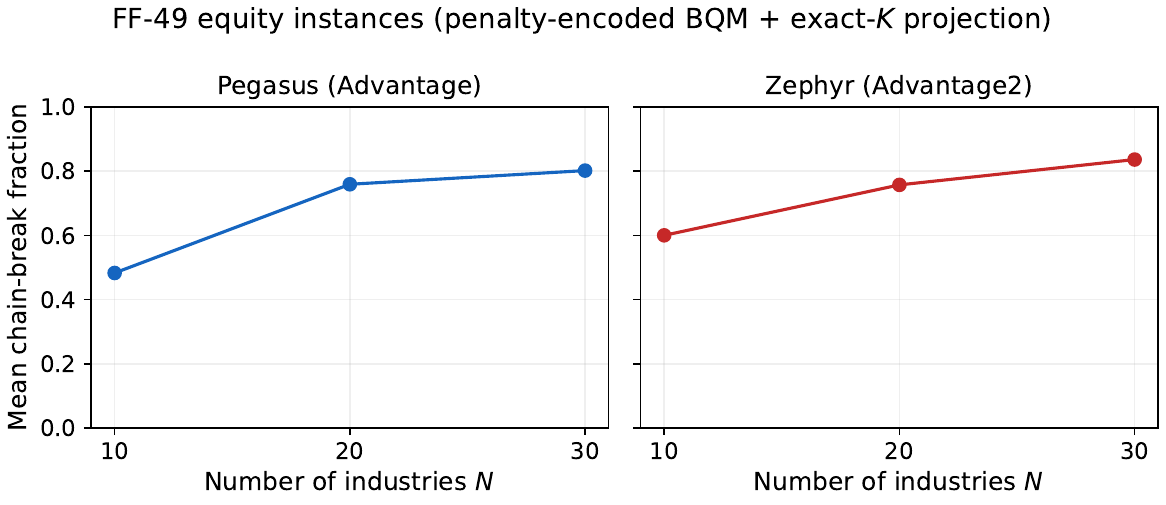}
\caption{Direct QPU on real Fama--French 49 equity data: chain-break fraction vs $N$ on Pegasus
and Zephyr, confirming the synthetic-data pattern on real financial covariance matrices.}
\label{fig:equity_sanity}
\end{figure}

\FloatBarrier
\subsection{Budget Sweep: Does More Time Help?}
\label{sec:b2}

Table~\ref{tab:b2} reports hybrid objective values across wall-clock budgets from 5 to 300
seconds, and Figure~\ref{fig:budget} visualizes the corresponding budget response curves. The
budget sweep uses $N \in \{50, 100, 200, 400\}$ (differing from the main $N$-sweep grid to
include mid-range sizes relevant for budget sensitivity).

\begin{table}[ht]
\centering
\caption{Hybrid objective value by wall-clock budget on the fixed
generator-seed-0 instance in each cell (dense family; block family shows the same
pattern, with two hybrid BQM block-family cells at $N{=}400$ at $180\,$s and $300\,$s not
completed within the campaign's wall-clock budget; the corresponding CQM cells did
complete). On the tested instances, CQM returns exactly identical recorded objective values at all
budgets; BQM often attains lower objectives at longer budgets but is
non-monotonic and does not reach the CQM objective at any budget.}
\label{tab:b2}
\begin{tabular}{rlrrrrr}
\toprule
$N$ & Solver & 5\,s & 30\,s & 60\,s & 180\,s & 300\,s \\
\midrule
50 & BQM & $-0.633$ & $-0.669$ & $-0.709$ & $-0.696$ & $-0.734$ \\
50 & CQM & $-0.799$ & $-0.799$ & $-0.799$ & $-0.799$ & $-0.799$ \\
\addlinespace
100 & BQM & $-1.309$ & $-1.367$ & $-1.593$ & $-1.792$ & $-1.695$ \\
100 & CQM & $-1.839$ & $-1.839$ & $-1.839$ & $-1.839$ & $-1.839$ \\
\addlinespace
200 & BQM & $-2.420$ & $-2.460$ & $-2.423$ & $-2.446$ & $-2.480$ \\
200 & CQM & $-3.598$ & $-3.598$ & $-3.598$ & $-3.598$ & $-3.598$ \\
\addlinespace
400 & BQM & $-4.426$ & $-4.921$ & $-5.080$ & $-5.740$ & $-5.899$ \\
400 & CQM & $-7.454$ & $-7.454$ & $-7.454$ & $-7.454$ & $-7.454$ \\
\bottomrule
\end{tabular}
\end{table}

\begin{figure}[ht]
\centering
\includegraphics[width=\textwidth]{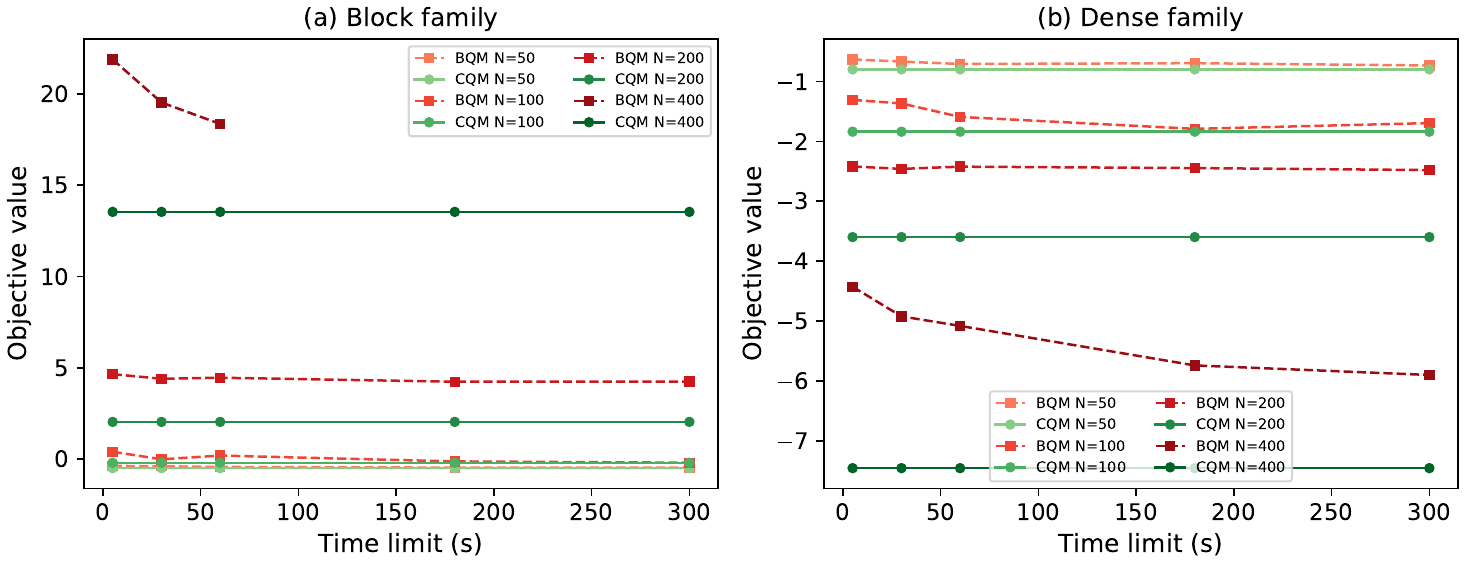}
\caption{Budget response curves for hybrid BQM (dashed) and CQM (solid) across problem sizes.
CQM lines are flat; BQM traces are non-monotonic but tend downward with longer budgets and
remain above CQM at all tested budgets.}
\label{fig:budget}
\end{figure}

\paragraph{Stochastic validation.} To verify that the budget-saturation finding is not an
artifact of single-shot runs, we repeat the hybrid CQM and BQM solves on
the same fixed generator-seed-0 instances, with 10 calls planned for each solver at
$N \in \{50, 100, 200\}$, both density families, and budgets $\{5, 300\}$\,s (228 runs total;
the $(N{=}50, \text{dense}, 300\,\text{s})$ cell completed four
repetitions per solver within the campaign's wall-clock budget; the
remaining cells reached the planned ten). Hybrid CQM
produced exactly identical recorded objective values
across all available repetitions at every tested $(N, \rho, B)$ combination, at both
$5\,$s and $300\,$s, verified at the cell level against the saved JSONL records. BQM
shows measurable stochastic variance (std $0.01$--$0.12$) and improves with budget, but its
best run at $300\,$s never reaches the CQM objective at $5\,$s on any tested instance.
The formal $v_{\mathrm{service}}$ metric is reported only for the $11$ CQM
cells meeting its $n\geq10$ requirement; identity in the four-call cell is
descriptive rather than protocol-compliant.
Figure~\ref{fig:repeat_saturation} visualizes this contrast.

\begin{figure}[!ht]
\centering
\includegraphics[width=0.95\textwidth]{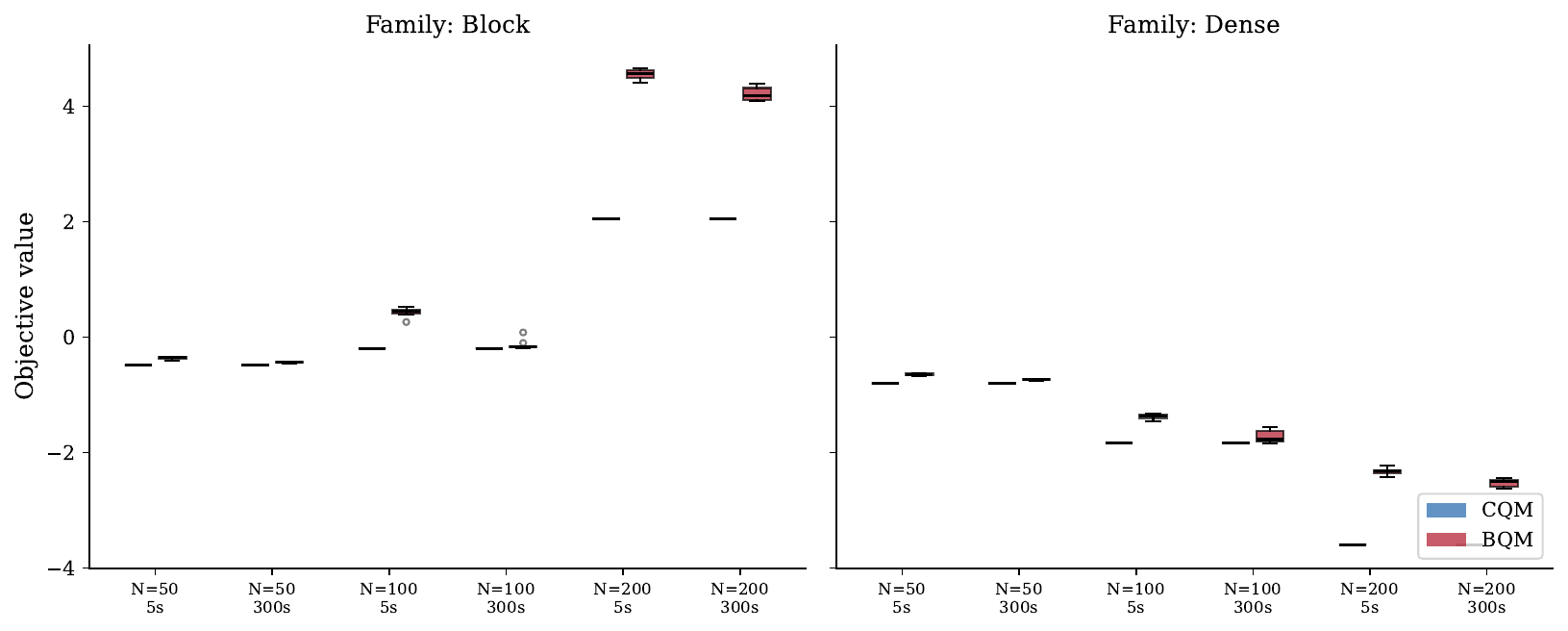}
\caption{Stochastic validation: box plots of objective values across up to 10 repeated calls for CQM
(blue, flat lines = zero variance) and BQM (red, visible spread) at budgets of 5\,s and
300\,s, split by density family.  Each solver has four calls in the
$(N{=}50,\text{dense},300\,\text{s})$ cell and ten elsewhere.}
\label{fig:repeat_saturation}
\end{figure}

\FloatBarrier
\subsection{QPU-Dwell vs Solution-Quality Correlation}
\label{sec:dwell-quality}

The QPU wall-clock fraction $r_{\mathrm{QPU}}$ bounds the time share
attributable to the quantum processor but cannot, by itself, bound the
QPU's causal contribution to solution quality.  We probe the causal
question with two complementary diagnostics, neither of which is itself
a causal ablation: within-cell association between
QPU dwell time and final objective (this subsection); and a CPU-only
matched-budget classical-reproducibility comparison
(Section~\ref{sec:qpu-replacement-ablation}).

\paragraph{Within-cell analysis.}
Within each $(N, \rho, B)$ cell, repeated hybrid CQM calls produce
exactly identical recorded objective values
(Section~\ref{sec:b2}), so within-cell correlation analysis is only
defined on hybrid BQM cells where service-internal randomness produces
measurable objective variance.  On $11$ such BQM cells (extracted from
the repeated-hybrid campaign at $N \in \{50, 100, 200\}$,
$\{$block, dense$\}$, $\{5, 300\}$\,s), we compute Spearman $\rho$
between \texttt{hybrid\_qpu\_access\_time} and
\texttt{objective\_value}, with $95\%$ bootstrap confidence intervals
(2{,}000 resamples).

The result is null after multiplicity correction.  Nine of $11$ cells
have percentile bootstrap CIs that straddle zero.  The two uncorrected
intervals excluding zero are $(N=100,\text{block},5\,\mathrm{s})$ with
$\rho=-0.78$ (raw $p=0.0075$, Holm $p=0.0755$) and
$(N=200,\text{dense},5\,\mathrm{s})$ with $\rho=+0.81$ (raw
$p=0.0049$, Holm $p=0.0535$); neither survives Holm correction across
the $11$ tests.  Their directions are also opposite.  The median
$|\rho|$ is $0.25$.  At the within-cell level, the repeated hybrid-BQM
records therefore show no multiplicity-controlled QPU-dwell--quality
association.

\paragraph{Cross-cell analysis at fixed $(\text{solver}, N)$.}
We also compute Spearman $\rho$ across the different $(\rho, B)$ cells
at fixed $(\text{solver}, N)$ in the budget sweep, where 10 distinct
cells per $(\text{solver}, N)$ vary in family and budget.  Of $8$
cross-cell buckets analysed, $2$ have negative associations that survive
Holm correction (BQM $N=400$: $\rho=-0.88$, adjusted $p=0.027$; CQM
$N=50$: $\rho=-0.93$, adjusted $p<0.001$), $0$ survive in the positive
direction, and $6$ do not reject zero.  The two negative cases may reflect the classical
decomposer allocating more QPU time to harder sub-problems within the
same $N$ (a hardness-correlated allocation pattern), but the
exposed telemetry cannot distinguish this allocation confounding from a
genuine contribution of the QPU samples; we report both cells rather
than explain them away.

Taken together, the dwell-vs-quality analysis finds no consistent
association and supports no causal inference on its own.
Section~\ref{sec:qpu-replacement-ablation} provides the complementary
matched-budget classical comparison.

\FloatBarrier
\subsection{Matched-Budget Classical Reproducibility Comparison}
\label{sec:qpu-replacement-ablation}

To test whether the objective levels reached by the hybrid service are
classically attainable at the same nominal budget, we performed a
CPU-only comparison: on $36$ problem instances
($N \in \{50, 80, 120, 200, 400, 640\}$, $\{$block, dense$\}$, three
instance seeds each), we ran three fixed algorithm seeds of
\texttt{TabuSampler} 1.7.0 from D-Wave's classical-samplers
package on the same penalty-encoded BQM that LeapHybridBQM consumes,
at a matched wall-clock budget of $5$\,s per instance (matching the
hybrid service's $5$\,s minimum, which is also the budget used by the
saved $N \in \{400, 640\}$ hybrid runs).  This gives $108$ Tabu runs;
each run uses one read and an independently seeded five-second search.
The runs execute serially on an Apple M3 Max CPU (16 cores, 128~GB RAM),
and every observed wall-clock is retained (mean $5.0025$\,s, range
$5.0005$--$5.0084$\,s across the $108$ runs).  Because timeout-based termination
can introduce minor rerun differences even at a fixed pseudorandom seed, we
treat the three algorithm seeds as stochastic replicates rather than
bitwise-reproducible fixtures.  The comparison grid starts at
$N = 50$ because hybrid CQM and BQM tie at $N = 10$ on every seed
(Table~\ref{tab:a2a3}) and the head-to-head gap is small through
$N = 30$, leaving little signal at the smaller sizes.
Results were projected to exact-$K$ feasibility using the same greedy
projector~\citep{lozano2026penalty} as the direct-QPU path, although all
$108$ Tabu outputs were already exact-$K$ feasible; the final
single-period mean-variance objective was then evaluated.

Table~\ref{tab:ablation} reports the resulting Tabu objectives
alongside the saved hybrid CQM, hybrid BQM, and Gurobi-optimal
objectives for the same $(N, \text{family}, \text{seed})$ instances.

\begin{table}[ht]
\centering
\caption{Matched-budget classical-reproducibility comparison:
\texttt{TabuSampler} at matched
$5$\,s wall-clock on the penalty-encoded BQM, compared to the saved
hybrid and Gurobi objectives on the same $(N, \text{family}, \text{seed})$
instances ($36$ problem instances, $108$ Tabu runs total). Tabu entries
are means over $3$ instance seeds $\times$ $3$ algorithm seeds per
$(N, \text{family})$ cell; lower is
better. The rightmost column is the Tabu--CQM delta (positive means
Tabu worse than hybrid CQM).  Gurobi anchors exist only at
$N \leq 200$ (academic-license model-size limit).}
\label{tab:ablation}
\begin{tabular}{rlrrrrr}
\toprule
$N$ & Family & Tabu & Hybrid CQM & Hybrid BQM & Gurobi$^\star$ & Tabu $-$ CQM \\
\midrule
 50 & block & $-0.4729$ & $-0.4746$ & $-0.3555$ & $-0.4746$ & $+0.0017$ \\
 50 & dense & $-0.7986$ & $-0.7986$ & $-0.6736$ & $-0.7986$ & $+0.0000$ \\
 80 & block & $-0.3088$ & $-0.3136$ & $+0.0691$ & $-0.3136$ & $+0.0049$ \\
 80 & dense & $-1.3721$ & $-1.3721$ & $-1.0242$ & $-1.3721$ & $+0.0000$ \\
120 & block & $+0.0093$ & $+0.0075$ & $+0.8249$ & $+0.0075$ & $+0.0018$ \\
120 & dense & $-2.0716$ & $-2.0716$ & $-1.4535$ & $-2.0716$ & $+0.0000$ \\
200 & block & $+1.9404$ & $+1.9395$ & $+4.3947$ & $+1.9395$ & $+0.0009$ \\
200 & dense & $-3.5833$ & $-3.5833$ & $-2.3086$ & $-3.6225^{\dagger}$ & $+0.0000$ \\
\midrule
400 & block & $+13.2036$ & $+13.2035$ & $+21.6853$ & --- & $+0.0001$ \\
400 & dense & $-7.2523$ & $-7.2523$ & $-4.4528$ & --- & $+0.0000$ \\
640 & block & $+39.4742$ & $+39.4742$ & $+61.0839$ & --- & $+0.0000$ \\
640 & dense & $-11.2615$ & $-11.2615$ & $-6.1166$ & --- & $+0.0000$ \\
\bottomrule
\end{tabular}
\\[2pt]
\footnotesize $^\star$ Gurobi proves optimality on all reported cells
except $(N=200, \text{dense}, \text{seed}=1)$, marked $\dagger$, which
hit \texttt{status\_9} in the academic-license run.
\end{table}

Aggregated across the $72$ seeded Tabu runs on the $24$ instances at
$N \leq 200$:
\begin{itemize}
  \item Mean absolute Tabu--CQM delta $=0.00116$; maximum $=0.00798$.
  $53/72$ runs, spanning $17/24$ instances for all three algorithm seeds,
  are within $0.0005$.
  \item Against the $23$ Gurobi-proved instances, mean absolute delta
  $=0.00121$ over $69$ seeded runs; maximum $=0.00798$, with $50/69$
  runs within $0.0005$.
  \item Mean(Tabu $-$ Hybrid BQM) $= -0.76635$.  Tabu substantially outperforms LeapHybridBQM on the same penalty-encoded inputs, so the hybrid-BQM degradation reflects the interaction of the BQM encoding with the LeapHybridBQM pipeline, not quantum sampling alone.
\end{itemize}

At $N \in \{400, 640\}$, the scales with the largest CQM-over-BQM
margins and no Gurobi anchor, the result sharpens: over $36$ seeded
runs on $12$ instances, mean absolute Tabu--CQM delta is
$3.04\times10^{-5}$ and the maximum is $0.000989$; $35/36$ runs are
within $0.0005$, and every Tabu solution is already exact-$K$ feasible
before projection. Tabu outperforms hybrid BQM by $9.51$ objective units
on average (maximum $21.69$).  The
classical-reproducibility finding therefore covers the full tested
range $N \leq 640$, not only the Gurobi-provable regime.

\paragraph{Interpretation.}
At matched $5$\,s wall-clock, a strong classical heuristic on the
penalty-encoded BQM produces objectives numerically close
from those produced by the hybrid CQM service on the same instances
(median absolute Tabu $-$ CQM gap $=0$ across all $108$ seeded runs,
mean absolute gap $0.00078$, and maximum $0.00798$).  Because the penalty-encoded BQM and the
constraint-native CQM encode the same cardinality-constrained MV
problem studied here, recovering Gurobi-optimal (or CQM-matching) objectives on the
BQM matches the solution quality of the CQM path; the encoding differs,
but the feasible set and the optimum do not.  (Equal objective value
does not by itself imply the identical selected portfolio unless the
optimizer is unique, which we do not claim.)  This is
\emph{not} a LeapHybridCQM internal ablation: the hybrid CQM pipeline
consumes a different problem encoding and we do not have access to its
proprietary internals.  What the result \emph{does} show is that the
objective levels reached by hybrid CQM on this problem class are
classically attainable at the same wall-clock budget on
the penalty-encoded variant, at every tested scale.  Combined with the
$\approx 0.7\%$ measured QPU wall-clock fraction
(Section~\ref{sec:timing-breakdown}), the experiment provides no positive
evidence for a distinct quantum-sampling contribution to solution quality;
it neither estimates nor rules out the QPU's contribution inside the
unopened CQM pipeline.

Tabu's substantial improvement over LeapHybridBQM (mean $\sim 0.77$ in
absolute objective at $N \leq 200$, growing to $\sim 9.5$ at
$N \in \{400, 640\}$) is consistent with the LeapHybridBQM pipeline
being limited by the interaction between its penalty-encoded input
and its specific decomposition policy.  The observed objective levels do
not require a uniquely quantum capability, but the external comparison does
not reveal how internal QPU samples are used.  By contrast the
constraint-native LeapHybridCQM pipeline, which avoids the penalty
term $A \cdot \mathbf{1} \mathbf{1}^\top$ that fully connects the
encoded logical graph (Theorem~\ref{thm:density_collapse}), reaches the
same objective levels as a strong classical heuristic alone.

\FloatBarrier
\subsection{Descriptive FF-49 Financial Overlay}
\label{sec:ff49-financial}

The synthetic-instance results
(Sections~\ref{sec:a1}--\ref{sec:b2}) measure solver behaviour on
controlled QUBOs; they do not directly answer the question a
quantitative portfolio manager will ask: whether the resulting
portfolios are competitive against simple financial baselines.  We
complement the solver-side analysis with an out-of-sample financial
evaluation on the Fama--French $49$-industry daily-return data set used
in the equity-data sanity check of Section~\ref{sec:a4}.

\paragraph{Reconstruction.}
For each of the $90$ post-projection portfolios saved from the direct-QPU
equity-data campaign, which comprises $5$ rolling rebalance windows
spanning May~1927 to December~2025, each with $N \in \{10, 20, 30\}$, two
topologies (Pegasus, Zephyr), and three chain strengths, we
reconstruct the daily out-of-sample portfolio return as the
equal-weight mean over the selected industries' realised daily returns
for the calendar month following the rebalance date.  The five
rebalance dates are chosen as evenly spaced indices over the $1{,}184$
available monthly rebalance dates in $1926$--$2025$.  No daily-return
data is required from the saved JSONL; only the saved selected-industries
list and the public Kenneth French daily data~\citep{french_data}.

\paragraph{Metrics.}
We report annualised Sharpe ratio (scaled by $\sqrt{252}$) and maximum
drawdown.  We do not attach inferential significance to either metric:
each evaluation month contains only 21--26 daily observations, the six
hardware configurations within a $(\text{window},N)$ cell are dependent,
and five sampled months are insufficient for a defensible probabilistic or
deflated-Sharpe analysis.

For comparison we compute the same metrics on the equal-weight all-$49$
industry portfolio (the $1/N$ baseline of
\citealp{demiguel2009optimal}) over the same five evaluation months.
The $1/N$ baseline is widely understood to be a hard-to-beat benchmark
in industry-level rolling allocation.

\paragraph{Aggregate result.}
Across all $90$ post-projection direct-QPU-path portfolios, the mean annualised Sharpe ratio is
$+1.94$; on the same five evaluation months, the $1/N$ all-$49$
baseline delivers a mean Sharpe of $+2.22$.  The direct-QPU-path
portfolios \emph{underperform} the $1/N$ baseline by a Sharpe delta of
$-0.28$ on average.

This finding is consistent with the broad
\citet{demiguel2009optimal} result that $1/N$ is hard to beat with
plug-in mean-variance allocations.  The five-window breakdown
(Table~\ref{tab:ff49-financial}) is descriptive only, because $5$ rolling
windows is too thin a sample for inferential claims about regime
dependence, but it shows that the average underperformance is not
uniform across the sampled period: the largest positive spread relative to
$1/N$ occurs in window~$0$ ($1927\text{-}06\text{-}01$ through
$1927\text{-}06\text{-}30$, where the post-projection $N{=}10$
direct-QPU-path portfolio
achieves Sharpe $+1.29$ while the $1/N$ baseline delivers $-1.21$),
and additional positive cell-level spreads occur in windows~$1$ and~$4$.
Those cells do not overturn the negative aggregate mean across all $90$
post-projection portfolios.

\begin{table}[ht]
\centering
\caption{Out-of-sample annualised Sharpe ratio on the FF-49 evaluation
months.  ``Mean QPU path'' is averaged over $6$ post-projection configurations per
$(\text{window}, N)$ cell ($2$ topologies $\times$ $3$ chain strengths).
``$1/N$ baseline'' is equal-weight over all 49 industries.}
\label{tab:ff49-financial}
\begin{tabular}{lrrrr}
\toprule
Window (rebalance date) & Path $N=10$ & Path $N=20$ & Path $N=30$ & $1/N$ \\
\midrule
0 ($1927\text{-}05\text{-}31$) & $+1.29$ & $+0.06$ & $-0.98$ & $-1.21$ \\
1 ($1951\text{-}12\text{-}31$) & $+2.05$ & $+2.64$ & $+2.19$ & $+2.20$ \\
2 ($1976\text{-}08\text{-}31$) & $+0.72$ & $+2.32$ & $+2.07$ & $+2.92$ \\
3 ($2001\text{-}04\text{-}30$) & $-1.98$ & $+0.74$ & $+1.36$ & $+1.85$ \\
4 ($2025\text{-}12\text{-}31$) & $+4.05$ & $+5.33$ & $+7.29$ & $+5.35$ \\
\addlinespace
Mean over windows                & $+1.23$ & $+2.22$ & $+2.39$ & $+2.22$ \\
\bottomrule
\end{tabular}
\end{table}

\paragraph{Deployment implication.}
The cardinality-constrained mean-variance problem
studied in this paper is solver-tractable on almost all anchored instances
(Gurobi proves optimum on $62$ of $63$ attempted cases at $N \leq 200$;
CQM matches Gurobi on all $54$ head-to-head cases at $N \leq 120$).
Separately, the evaluated post-projection direct-QPU-path portfolios are not financially
competitive against $1/N$ at the aggregate level on the five FF-49
evaluation months.
The practical recommendation to a portfolio manager evaluating D-Wave
hybrid solvers is therefore double-edged: the CQM path matches all $54$
available Gurobi-proved head-to-head objectives and its objective levels are
also reached by a classical Tabu heuristic at matched compute
(Sections~\ref{sec:timing-breakdown}--\ref{sec:qpu-replacement-ablation});
the post-projection direct-QPU-path portfolios reconstructed in this overlay underperform
$1/N$ in Sharpe terms on average across the five sampled windows.
Because CQM decision vectors were not retained in the equity campaign,
the overlay does not establish out-of-sample financial performance for
the CQM path.  For a rebalancing system
with a $5$\,s decision cadence, the relevant operational question is
not the unidentifiable internal allocation alone, but whether the resulting
portfolios are worth the entire pipeline's operational and cloud cost;
on the present data, the answer on the aggregate of the five evaluated
windows is no.

\FloatBarrier
% ============================================================================
\section{Discussion}
\label{sec:discussion}
% ============================================================================

\subsection{Formulation--Solver Interaction}

On the tested instance families, the constraint-native CQM service produces lower objective
values than the hybrid-BQM service at every $N > 10$, across all three density
families, and at every tested budget. One relevant mechanism is that the penalty term $A \cdot
\mathbf{1}\mathbf{1}^\top$ adds $O(N^2)$ couplers with magnitude $A$ to the QUBO, while the
economic signal in $\Sigma$ is distributed across entries whose individual magnitudes diminish
as $N$ grows. This is consistent with the constraint-dilution diagnosis
of~\citet{lozano2026penalty}.  It is not a formulation-only explanation:
Tabu reaches CQM-level objectives on the same BQM, so the observed degradation
depends on the interaction between encoding and solver pipeline.

\subsection{Density Collapse as Structural Explanation}

The density-axis collapse (Section~\ref{sec:density_collapse}) provides a structural explanation
for the formulation gap. For the tested direct-QPU and hybrid-BQM encodings, the penalty term produces a complete
logical graph regardless of the original $\Sigma$ density. Chain-break fractions and embedding
overheads at each $N$ are virtually identical across diagonal, block, and dense families. For
CQM, the original density is preserved, potentially allowing the hybrid solver's classical
component to exploit sparsity. This asymmetry means that density is a meaningful benchmark axis
only for the CQM path.

\subsection{CQM Budget Saturation}

On the tested instances, the CQM solver returns identical recorded
objective values at 5 and 300 seconds.  This shows that longer requested
budgets do not improve the retained objective on these fixed instances;
it does not reveal the internal trajectory, prove convergence at five
seconds, or imply the same behaviour for richer constraints or
non-binary variables.

\subsection{Role of Direct QPU}

Direct QPU access with penalty-encoded QUBOs incurs structural embedding
overhead from the complete logical graph. At $N = 80$, the tested
embeddings use $750$ physical qubits on average ($9.4\times$ overhead),
while the tested, untuned chain-strength grid produces chain-break fractions
above $0.93$. For practitioners choosing within the D-Wave ecosystem, hybrid
CQM offers better solution quality without user-managed embedding or
exact-$K$ projection. Direct QPU remains relevant for studying annealing
dynamics and for problem classes with naturally sparse logical graphs.

\subsection{Implications for Quantum Finance Benchmarks and Deployment}

These results motivate more qualified reporting in quantum-finance benchmarks. The CQM path
matches Gurobi on the tested instances, but Gurobi also proves optimality, so the result is a
practical D-Wave deployment finding rather than a classical-superiority finding. The field-level
implication is that quantum-finance evaluations can be dominated by modeling-interface choices:
the hybrid BQM path degrades under penalty encoding, whereas the CQM path returns
Gurobi-optimal objectives on the anchored grid. Future quantum-finance claims should therefore separate
three questions: modeling fidelity (does the formulation faithfully represent the financial
problem?), hybrid-solver usefulness (does the D-Wave workflow produce good solutions within
operational time budgets?), and classical competitiveness (does the quantum-assisted path
outperform strong classical alternatives?).

For financial institutions evaluating D-Wave, the immediate recommendation is to prototype
constrained portfolio and allocation problems in CQM first, measure against a strong MIQP
baseline, and avoid treating QUBO penalty tuning as the default production route. This is
especially relevant to prediction-market optimization systems and market venues where portfolio
selection over bets or contracts combines cardinality, exposure limits, budget constraints, and
turnover costs.

A second deployment-relevant consequence comes from the budget-saturation result of
Section~\ref{sec:b2}. For periodic-rebalancing systems with a decision cadence on the order of
the hybrid solver's 5\,s minimum or longer (intraday sector-rotation engines, end-of-period
portfolio re-allocation, batch contract-basket re-sizing, or any periodic re-optimization over
many simultaneously open events or contracts), the relevant operational metric is not per-call
quality at unbounded budget, but the quality reachable within the shortest acceptable budget.
On the tested instances the CQM service returns the same recorded objective
value at the 5\,s minimum and at every longer budget up to 300\,s, and the
same objective value across repeated calls. The saved records do not establish
that every corresponding decision vector or internal trajectory is identical. In a system
where the rebalancing rate, not per-call solve time, is the binding constraint, the
orchestrator's decision collapses to ``how often to re-solve?'' rather than ``how long to
wait?'', and budget tuning ceases to be a system-level parameter. We deliberately exclude
latency-critical execution paths (sub-second in-play markets, microsecond execution venues)
from this discussion: a 5\,s minimum solve time is incompatible with those use cases regardless
of solution quality. The penalty-encoded BQM path admits no analogous guarantee on the tested
instances at any tested budget.

\subsection{Scope and Limitations}

This study compares formulations and solver interfaces within the D-Wave platform. It is not a
general quantum-vs-classical benchmark: Gurobi MIQP proves optimality for
$62$ of $63$ attempted instances at $N \leq 200$ within 300 seconds, and we
make no claim that any D-Wave path is competitive with
state-of-the-art classical solvers on this problem class.

Additional limitations:
\begin{itemize}
    \item Every empirical campaign uses the single-period
    $\boldsymbol{\tau}=\mathbf{0}$ special case.  The optional turnover
    extension is derived but non-zero switching costs and path-dependent
    rebalancing are not tested.
    \item The penalty weight $A = 4.0$ is held fixed following the calibration
    of~\citet{lozano2026penalty}; Appendix~\ref{app:penalty} verifies robustness to $A \in
    \{2, 4, 8\}$ but does not perform a full per-instance sweep.
    \item The real-data section uses only $N \leq 30$; it serves as a directional sanity check,
    not as the primary evidence.
    \item Warm-start / reverse-annealing experiments are deferred to future work.
    \item Richer constraint sets (sector exposure caps, ESG exclusions, mutual exclusivity,
    and multi-period turnover budgets) would test whether the CQM advantage persists under
    institutional-grade complexity. Prediction-market and sports-betting allocation, where a
    trader selects a limited number of contracts subject to budget and rebalancing costs,
    provides a natural sparse-graph extension.
    \item Stochastic validation (up to 10 repeated calls) is performed on
    a representative fixed-seed subset; one cell has four calls per
    solver, and the full grid is single-shot per configuration.
    \item The direct-QPU chain-strength grid spans $25\%$--$100\%$ of the
    penalty-only pairwise Ising scale but does not include Ocean's
    torque-compensated default or a larger scale-aware sweep; it also uses
    automatic scaling, one gauge, and one anneal duration.  Its chain-break
    results characterize that protocol, not a tuned lower bound on direct-QPU
    performance.
\end{itemize}

% ============================================================================
\section{Conclusion}
\label{sec:conclusion}
% ============================================================================

We have audited where QPU access appears in the exposed telemetry of D-Wave's
hybrid quantum-classical service for cardinality-constrained portfolio
selection.  The empirical campaign uses the single-period
$\boldsymbol{\tau}=\mathbf{0}$ objective; the optional turnover extension
is not tested. The headline observation is operational: the constraint-native LeapHybridCQM service
matches Gurobi's proven optimum on all $54$ head-to-head instances at $N \leq 120$, but the
mean QPU access time is $0.034$\,s: $0.68\%$ of the nominal $5$\,s
budget and approximately $0.72\%$ of measured service run time. The
remaining $\sim 99\%$ of wall-clock is not
labelled QPU access; the exposed telemetry does not resolve that residual into classical
computation versus orchestration or minimum-budget padding. At the level of
reported service telemetry, the successful CQM path is dominated by non-QPU
service time, consistent with the service's documented
classically-orchestrated design~\citep{dwave_hybrid_solvers_whitepaper}, with a small
measured QPU-access component;
the CPU-only comparison of Section~\ref{sec:qpu-replacement-ablation} shows those objective
levels are classically attainable at matched wall-clock: mean absolute delta
$0.00078$ across $108$ fixed-seed runs on $36$ instances (maximum $0.00798$),
falling to $3.04\times10^{-5}$ at $N \in \{400,640\}$ (maximum $0.000989$).
This is an objective-achievability result, not a bound on
the QPU's causal role inside the unopened pipeline. A true internal sub-solver ablation of the LeapHybridCQM
pipeline, replacing the QPU sub-call with a classical sub-solver inside D-Wave's proprietary
decomposition workflow, remains future work requiring vendor cooperation. Three further
results sharpen this audit:

\begin{enumerate}
    \item At the tested coefficients, the cardinality penalty contributes
    a dense rank-one matrix that makes the penalty-encoded logical graph
    fully connected regardless of covariance density, collapsing the
    density benchmark axis for BQM and direct-QPU paths.  Theorem~\ref{thm:density_collapse}
    states the finite exact-cancellation exception explicitly.
    \item The CQM service returns exactly identical recorded objective values at every tested
    wall-clock budget from $5$ to $300$ seconds and across repeated calls; the formal
    $v_{\mathrm{service}}$ result excludes the one four-call cell.  On these fixed instances,
    requesting a longer
    budget provides no retained-objective improvement.
    \item Constraint-native CQM produces lower objective values than the penalty-encoded BQM
    interface at every tested $N > 10$ across all density families.  Simulated annealing
    also degrades on the BQM, while seeded Tabu runs largely reproduce CQM-level
    objectives; together these results identify a formulation--pipeline interaction rather
    than a solver-independent failure.  A paired one-sided Wilcoxon test rejects
    $H_0: f_{\mathrm{CQM}} \geq f_{\mathrm{BQM}}$ at every $N \geq 20$ after Holm correction
    ($p_{\mathrm{Holm}} < 0.05$).
\end{enumerate}

For practitioners using D-Wave hybrid solvers, these results recommend the constraint-native
CQM interface and, more importantly, change how D-Wave hybrid performance should be
reported. A reported "hybrid win" carries no automatic implication that the QPU is doing
material work; the QPU's operational share must be measured and disclosed alongside the
wall-clock and quality numbers. We propose that future quantum-finance benchmarks report
\texttt{run\_time}, \texttt{charge\_time}, and \texttt{qpu\_access\_time} separately, together
with the QPU-fraction of wall-clock, so that resource accounting is transparent while causal
attribution, which telemetry alone cannot settle, remains an explicitly separate
question.

\subsection*{Data and Code Availability}

Synthetic instances are generated from parametric covariance models with fixed random seeds for
reproducibility. Equity returns are from the Kenneth French Data Library, 49 industry
portfolios, daily~\citep{french_data}. Source code for all experiments (QUBO/CQM construction,
solver wrappers, analysis scripts, audit-protocol reference implementation), random seeds,
solver configurations, and processed result JSONLs are publicly available under the MIT License
at \url{https://github.com/LuisLozanoM/where-quantum-lives}. The Wilcoxon,
dwell-vs-quality, timing-distribution, seeded-Tabu, and repeated-call
cell-level verification scripts referenced in this paper live under
\texttt{analysis/} in that repository.

\section*{Acknowledgments}

I thank the D-Wave team for providing access to quantum computing resources through the Leap
cloud platform and for their continuous technical support. Live QPU experiments used
\texttt{Advantage\_system4.1} (Pegasus) and \texttt{Advantage2\_system1.13} (Zephyr).

\subsection*{Use of Generative Artificial Intelligence}

In accordance with IOP Publishing's policy, I disclose that this manuscript
and its reproducibility code were prepared with assistance from the following
generative-AI tools during April--July 2026:
\begin{itemize}\itemsep=2pt
  \item \textbf{Anthropic Claude Opus 4.7} via Claude Code, used for
  analysis-code scaffolding, LaTeX drafting, copyediting, and adversarial review.
  \item \textbf{OpenAI GPT-5} via Codex command-line and desktop interfaces,
  used for independent code review, CPU-only verification and statistical
  analysis, manuscript critique, and copyediting.
  \item \textbf{Moonshot Kimi~K2}, used for independent manuscript critique.
  \item \textbf{OpenAI ChatGPT (GPT-5 Pro)}, used for exploratory
  brainstorming and alternative-framing discussions.
\end{itemize}

These tools did not fabricate, replace, or alter the original QPU or hybrid
service records.  They assisted with code generation, execution of reproducibility
checks, interpretation checks, and prose revision.  I inspected the code and
outputs, verified the numerical claims against the saved records, approved every
sentence, and made all scientific and submission decisions.  I take full
responsibility for the manuscript.

\section*{Statements and Declarations}

\subsection*{Competing Interests}

I declare no competing interests.

\subsection*{Funding}

This research received no external funding. QPU access was obtained through my personal
D-Wave Leap cloud subscription.

\subsection*{Authors' Contributions}

I am the sole author and conceived the study, designed the experiments, implemented the
code, ran the experiments, analyzed the results, and wrote the manuscript.

\appendix

\section{Penalty Robustness Check}
\label{app:penalty}

The penalty weight $A = 4.0$ used in the main experiments is consistent with the calibration
study of~\citet{lozano2026penalty}. We verify that the BQM results are not artifacts of a
single penalty value by running the hybrid BQM solver at $A \in \{2, 4, 8\}$ on a subset of 12
instances ($N \in \{20, 50, 80, 120\}$, all three density families, seed~0).

Table~\ref{tab:penalty} reports the projected objective values. All 36 runs returned
raw-feasible solutions (the hybrid solver satisfied the penalty-encoded cardinality constraint
before post-processing in every case). The objective values are qualitatively similar across
penalty values at each $(N, \rho)$, with no systematic trend favoring any particular $A$.
The hybrid-BQM degradation at large $N$ is present at all three tested penalty
levels, supporting robustness over this limited grid.  It does not exclude a
better per-instance penalty or solver configuration.

\begin{table}[!ht]
\centering
\caption{Hybrid BQM objective by penalty weight $A$ (projected to exact-$K$). All 36 runs
returned raw-feasible solutions. Objective values are qualitatively stable across $A$.}
\label{tab:penalty}
\small
\begin{tabular}{rlrrr}
\toprule
$N$ & Family & $A = 2$ & $A = 4$ & $A = 8$ \\
\midrule
20 & diagonal & $-0.340$ & $-0.340$ & $-0.340$ \\
50 & diagonal & $-0.587$ & $-0.603$ & $-0.590$ \\
80 & diagonal & $-0.929$ & $-0.931$ & $-0.910$ \\
120 & diagonal & $-1.239$ & $-1.155$ & $-1.181$ \\
\addlinespace
20 & block & $-0.252$ & $-0.252$ & $-0.264$ \\
50 & block & $-0.348$ & $-0.361$ & $-0.395$ \\
80 & block & $-0.059$ & $-0.069$ & $-0.040$ \\
120 & block & $+0.909$ & $+1.003$ & $+0.982$ \\
\addlinespace
20 & dense & $-0.365$ & $-0.365$ & $-0.350$ \\
50 & dense & $-0.672$ & $-0.677$ & $-0.693$ \\
80 & dense & $-1.119$ & $-1.093$ & $-1.060$ \\
120 & dense & $-1.613$ & $-1.462$ & $-1.453$ \\
\bottomrule
\end{tabular}
\end{table}

% ============================================================================
% References
% ============================================================================

% Flush pending appendix floats before the bibliography without forcing an
% otherwise empty float page.
\FloatBarrier

\bibliographystyle{plainnat}
\begingroup
\small
\bibliography{references}

\begin{thebibliography}{24}
\providecommand{\natexlab}[1]{#1}
\providecommand{\url}[1]{\texttt{#1}}
\expandafter\ifx\csname urlstyle\endcsname\relax
  \providecommand{\doi}[1]{doi: #1}\else
  \providecommand{\doi}{doi: \begingroup \urlstyle{rm}\Url}\fi

\bibitem[Acharya et~al.(2025)Acharya, Yalovetzky, Minssen, Chakrabarti,
  Shaydulin, Raymond, Sun, Herman, Andrist, Salton, Schuetz, Katzgraber, and
  Pistoia]{acharya2025}
Atithi Acharya, Romina Yalovetzky, Pierre Minssen, Shouvanik Chakrabarti,
  Ruslan Shaydulin, Rudy Raymond, Yue Sun, Dylan Herman, Ruben~S. Andrist,
  Grant Salton, Martin J.~A. Schuetz, Helmut~G. Katzgraber, and Marco Pistoia.
\newblock Decomposition pipeline for large-scale portfolio optimization with
  applications to near-term quantum computing.
\newblock \emph{Physical Review Research}, 7:\penalty0 023142, 2025.
\newblock \doi{10.1103/PhysRevResearch.7.023142}.

\bibitem[Acuaviva et~al.(2026)Acuaviva, Aguirre, Pe{\~n}a, and
  Sanz]{acuaviva2026}
Arturo Acuaviva, David Aguirre, Rub{\'e}n Pe{\~n}a, and Mikel Sanz.
\newblock Benchmarking quantum computers: Towards a standard performance
  evaluation approach.
\newblock \emph{Quantum Science and Technology}, 11\penalty0 (2):\penalty0
  025004, 2026.
\newblock \doi{10.1088/2058-9565/ae4269}.

\bibitem[Bertsimas and Shioda(2009)]{bertsimas2009cardinality}
Dimitris Bertsimas and Romy Shioda.
\newblock Algorithm for cardinality-constrained quadratic optimization.
\newblock \emph{Computational Optimization and Applications}, 43\penalty0
  (1):\penalty0 1--22, 2009.
\newblock \doi{10.1007/s10589-007-9126-9}.

\bibitem[Boothby et~al.(2020)Boothby, Bunyk, Raymond, and
  Roy]{boothby-bunyk-raymond-roy-2020}
Kelly Boothby, Paul Bunyk, Jack Raymond, and Aidan Roy.
\newblock Next-generation topology of {D-Wave} quantum processors, 2020.
\newblock Pegasus and Zephyr topology design and embedding constants.

\bibitem[Bucher et~al.(2024)Bucher, Kraus, Blenninger, Lachner, Stein, and
  Linnhoff-Popien]{bucher2024}
David Bucher, Nico Kraus, Jonas Blenninger, Michael Lachner, Jonas Stein, and
  Claudia Linnhoff-Popien.
\newblock Towards robust benchmarking of quantum optimization algorithms.
\newblock In \emph{2024 IEEE International Conference on Quantum Computing and
  Engineering (QCE)}, pages 159--170, 2024.
\newblock \doi{10.1109/QCE60285.2024.11030870}.

\bibitem[Buonaiuto et~al.(2023)Buonaiuto, Gargiulo, De~Pietro, Esposito, and
  Pota]{buonaiuto2023}
Giuseppe Buonaiuto, Francesco Gargiulo, Giuseppe De~Pietro, Massimo Esposito,
  and Marco Pota.
\newblock Best practices for portfolio optimization by quantum computing,
  experimented on real quantum devices.
\newblock \emph{Scientific Reports}, 13:\penalty0 19434, 2023.
\newblock \doi{10.1038/s41598-023-45392-w}.

\bibitem[{D-Wave Systems}(2026{\natexlab{a}})]{dwave_hybrid}
{D-Wave Systems}.
\newblock Sampling with hybrid solvers, 2026{\natexlab{a}}.
\newblock URL
  \url{https://docs.dwavequantum.com/en/latest/industrial_optimization/index_hybrid.html}.
\newblock Accessed July 11, 2026.

\bibitem[{D-Wave Systems}(2026{\natexlab{b}})]{dwave_hybrid_solvers_whitepaper}
{D-Wave Systems}.
\newblock Sampling with hybrid solvers: Solver timing, 2026{\natexlab{b}}.
\newblock URL
  \url{https://docs.dwavequantum.com/en/latest/industrial_optimization/index_hybrid.html#solver-timing}.
\newblock Defines \texttt{run\_time}, \texttt{charge\_time}, and
  \texttt{qpu\_access\_time} for the LeapHybridCQM and LeapHybridBQM samplers;
  accessed 2026-07-11.

\bibitem[DeMiguel et~al.(2009)DeMiguel, Garlappi, and
  Uppal]{demiguel2009optimal}
Victor DeMiguel, Lorenzo Garlappi, and Raman Uppal.
\newblock Optimal versus naive diversification: How inefficient is the {$1/N$}
  portfolio strategy?
\newblock \emph{Review of Financial Studies}, 22\penalty0 (5):\penalty0
  1915--1953, 2009.
\newblock \doi{10.1093/rfs/hhm075}.

\bibitem[French(2026)]{french_data}
Kenneth~R. French.
\newblock 49 industry portfolios, 2026.
\newblock URL
  \url{https://mba.tuck.dartmouth.edu/pages/faculty/ken.french/Data_Library/det_49_ind_port.html}.
\newblock Accessed March 31, 2026.

\bibitem[Herman et~al.(2023)Herman, Googin, Liu, Sun, Galda, Safro, Pistoia,
  and Alexeev]{herman2023}
Dylan Herman, Cody Googin, Xiaoyuan Liu, Yue Sun, Alexey Galda, Ilya Safro,
  Marco Pistoia, and Yuri Alexeev.
\newblock Quantum computing for finance.
\newblock \emph{Nature Reviews Physics}, 5\penalty0 (8):\penalty0 450--465,
  2023.
\newblock \doi{10.1038/s42254-023-00603-1}.

\bibitem[Koch et~al.(2026)Koch, Bernal~Neira, Chen, et~al.]{koch2026qoblib}
Thorsten Koch, David~E. Bernal~Neira, Ying Chen, et~al.
\newblock The quantum optimization benchmarking library.
\newblock \emph{Nature Computational Science}, 6\penalty0 (6):\penalty0
  653--671, 2026.
\newblock \doi{10.1038/s43588-026-00991-1}.

\bibitem[Lang et~al.(2022)Lang, Zielinski, and Feld]{lang2022}
Jonas Lang, Sebastian Zielinski, and Sebastian Feld.
\newblock Strategic portfolio optimization using simulated, digital, and
  quantum annealing.
\newblock \emph{Applied Sciences}, 12\penalty0 (23):\penalty0 12288, 2022.
\newblock \doi{10.3390/app122312288}.

\bibitem[Lozano(2026)]{lozano2026penalty}
Luis Lozano.
\newblock A penalty-free pipeline for direct quantum-annealer portfolio
  optimization, 2026.
\newblock URL \url{https://arxiv.org/abs/2605.17628}.

\bibitem[Lubinski et~al.(2023)Lubinski, Johri, Varosy, Coleman, Zhao, Necaise,
  Baldwin, Mayer, and Proctor]{lubinski2023}
Thomas Lubinski, Sonika Johri, Paul Varosy, Jeremiah Coleman, Luning Zhao,
  Jason Necaise, Charles~H. Baldwin, Karl Mayer, and Timothy Proctor.
\newblock Application-oriented performance benchmarks for quantum computing.
\newblock \emph{IEEE Transactions on Quantum Engineering}, 4:\penalty0 1--32,
  2023.
\newblock \doi{10.1109/TQE.2023.3253761}.

\bibitem[Lucas(2014)]{lucas2014}
Andrew Lucas.
\newblock Ising formulations of many {NP} problems.
\newblock \emph{Frontiers in Physics}, 2:\penalty0 5, 2014.
\newblock \doi{10.3389/fphy.2014.00005}.

\bibitem[Markowitz(1952)]{markowitz1952}
Harry Markowitz.
\newblock Portfolio selection.
\newblock \emph{The Journal of Finance}, 7\penalty0 (1):\penalty0 77--91, 1952.
\newblock \doi{10.2307/2975974}.

\bibitem[Morapakula et~al.(2025)Morapakula, Deshpande, Yata, Ubale, Wad, and
  Ikeda]{morapakula2025endtoend}
Sai~Nandan Morapakula, Sangram Deshpande, Rakesh Yata, Rushikesh Ubale, Uday
  Wad, and Kazuki Ikeda.
\newblock End-to-end portfolio optimization with hybrid quantum annealing.
\newblock \emph{Advanced Quantum Technologies}, 2025.
\newblock \doi{10.1002/qute.202500753}.

\bibitem[Mugel et~al.(2022)Mugel, Kuchkovsky, S\'anchez, Fern\'andez-Lorenzo,
  Luis-Hita, Lizaso, and Or\'us]{mugel2022}
Samuel Mugel, Carlos Kuchkovsky, Escol\'astico S\'anchez, Samuel
  Fern\'andez-Lorenzo, Jorge Luis-Hita, Enrique Lizaso, and Rom\'an Or\'us.
\newblock Dynamic portfolio optimization with real datasets using quantum
  processors and quantum-inspired tensor networks.
\newblock \emph{Physical Review Research}, 4:\penalty0 013006, 2022.
\newblock \doi{10.1103/PhysRevResearch.4.013006}.

\bibitem[Or\'us et~al.(2019)Or\'us, Mugel, and Lizaso]{orus2019}
Rom\'an Or\'us, Samuel Mugel, and Enrique Lizaso.
\newblock Quantum computing for finance: Overview and prospects.
\newblock \emph{Reviews in Physics}, 4:\penalty0 100028, 2019.
\newblock \doi{10.1016/j.revip.2019.100028}.

\bibitem[Pelofske(2025)]{pelofske2025}
Elijah Pelofske.
\newblock Comparing three generations of {D-Wave} quantum annealers for minor
  embedded combinatorial optimization problems.
\newblock \emph{Quantum Science and Technology}, 10\penalty0 (2):\penalty0
  025025, 2025.
\newblock \doi{10.1088/2058-9565/adb029}.

\bibitem[Sakuler et~al.(2025)Sakuler, Oberreuter, Aiolfi, Asproni, Roman, and
  Schiefer]{sakuler2025}
Wolfgang Sakuler, Johannes~M. Oberreuter, Riccardo Aiolfi, Luca Asproni,
  Branislav Roman, and J{\"u}rgen Schiefer.
\newblock A real-world test of portfolio optimization with quantum annealing.
\newblock \emph{Quantum Machine Intelligence}, 7:\penalty0 43, 2025.
\newblock \doi{10.1007/s42484-025-00268-2}.

\bibitem[Venturelli and Kondratyev(2019)]{venturelli2019}
Davide Venturelli and Alexei Kondratyev.
\newblock Reverse quantum annealing approach to portfolio optimization
  problems.
\newblock \emph{Quantum Machine Intelligence}, 1:\penalty0 17--30, 2019.
\newblock \doi{10.1007/s42484-019-00001-w}.

\bibitem[Verma and Lewis(2022)]{verma2020}
Amit Verma and Mark Lewis.
\newblock Penalty and partitioning techniques to improve performance of {QUBO}
  solvers.
\newblock \emph{Discrete Optimization}, 44:\penalty0 100594, 2022.
\newblock \doi{10.1016/j.disopt.2020.100594}.

\end{thebibliography}
\endgroup

\end{document}